\documentclass[aap, preprint]{imsart}

\RequirePackage{amsthm,amsmath,amsfonts,amssymb}
\RequirePackage[numbers,sort&compress]{natbib}
\RequirePackage[colorlinks,citecolor=blue,urlcolor=blue]{hyperref}
\RequirePackage{graphicx}

\usepackage{lineno}

\usepackage{lipsum}
\startlocaldefs
\theoremstyle{plain}

\newtheorem{theorem}{Theorem}[section]
\newtheorem{lemma}[theorem]{Lemma}

\newtheorem{proposition}[theorem]{Proposition}
\newtheorem{corollary}[theorem]{Corollary}

\theoremstyle{remark}
\newtheorem{remark}{Remark}
\newtheorem{definition}[theorem]{Definition}
\newtheorem{example}{Example}

\endlocaldefs

\usepackage{color}
\usepackage{enumerate}
\usepackage{enumitem}
\usepackage{tikz}
\usepackage{pgfplots}
\usepackage{comment}
\usepackage{centernot}

\newcommand{\be}{\begin{equation}}
\newcommand{\ee}{\end{equation}}
\newcommand{\bea}{\begin{eqnarray}}
\newcommand{\eea}{\end{eqnarray}}
\newcommand{\barr}{\begin{array}}
\newcommand{\earr}{\end{array}}
\newcommand{\bn}{\begin{enumerate}}
\newcommand{\en}{\end{enumerate}}
\newcommand{\bi}{\begin{itemize}}
\newcommand{\ei}{\end{itemize}}
\newcommand{\bbbm}{\begin{pmatrix}}
\newcommand{\eeem}{\end{pmatrix}}

\newcommand{\la}{\lambda}
\newcommand{\La}{\Lambda}
\newcommand{\si}{\sigma}

\newcommand{\bbu}{{\bf u}}

\newcommand{\bbx}{{\bf x}}

\newcommand{\cC}{{\cal C}}
\newcommand{\cD}{{\cal D}}
\newcommand{\cL}{{\cal L}}
\newcommand{\cM}{{\cal M}}
\newcommand{\cN}{{\cal N}}
\newcommand{\cP}{{\cal P}}

\newcommand{\C}{{\mathbb C}}
\newcommand{\E}{{\mathbb E}}

\newcommand{\N}{{\mathbb N}}
\newcommand{\R}{{\mathbb R}}

\newcommand{\nn}{\nonumber}

\newcommand{\Wt}{{{\cal W}_{2}}}
\newcommand{\Wtq}{{{\cal W}_{2}^2}}

\newcommand{\spt}{\mathop{\rm spt}}

\newcommand*{\tr}{\mathsf{T}}

\DeclareMathOperator*{\argmin}{argmin}
\DeclareMathOperator*{\argmax}{argmax}
\DeclareMathOperator*{\Var}{{Var}}
\DeclareMathOperator{\id}{Id}

\newcommand{\K}{\preceq_{\rm K}}
\renewcommand{\C}{\preceq_{\rm C}}

  \definecolor{darkspringgreen}{rgb}{0.09, 0.45, 0.27} 
 \definecolor{darkgray}{rgb}{0.66, 0.66, 0.66}
 \definecolor{DeepBlue}{RGB}{0,0,160} % a strong dark blue

\begin{document}

\begin{frontmatter}
\title{
Data denoising with self consistency, variance maximization, and the Kantorovich dominance
}
%\\ (or)\\
%Variance Maximization Under Distributional Dominance for Data Denoising

%\title{A sample article title with some additional note\thanksref{t1}}
\runtitle{Variance Maximization Under Distributional Dominance for Data Denoising}
%\thankstext{T1}{A sample additional note to the title.}

\begin{aug}
%%%%%%%%%%%%%%%%%%%%%%%%%%%%%%%%%%%%%%%%%%%%%%%
%% Only one address is permitted per author. %%
%% Only division, organization and e-mail is %%
%% included in the address.                  %%
%% Additional information can be included in %%
\author[A]{\fnms{Joshua}~\snm{Zoen-Git Hiew}\ead[label=e3]{joshuazo@ualberta.ca}},
\author[B]{\fnms{Tongseok}~\snm{Lim}\thanks{All authors contributed equally to this work and are listed in alphabetical order.}\ead[label=e2]{lim336@purdue.edu}\orcid{0000-0002-1290-6964}},
\author[A]{\fnms{Brendan}~\snm{Pass}\ead[label=e1]{pass@ualberta.ca}},
\and
\author[C]{\fnms{Marcelo}~\snm{Cruz de Souza}\ead[label=e4]{marcelo.souza@bcb.gov.br}}
\
%%%%%%%%%%%%%%%%%%%%%%%%%%%%%%%%%%%%%%%%%%%%%%
%% Addresses                                %%
%%%%%%%%%%%%%%%%%%%%%%%%%%%%%%%%%%%%%%%%%%%%%%

\address[A]{Department of Mathematical and Statistical Sciences, University of Alberta, Edmonton AB Canada\printead[presep={ ,\ }]{e1,e3}}

\address[B]{Mitch Daniels School of Business, Purdue University, West Lafayette, IN 47907, USA \printead[presep={,\ }]{e2}}

\address[C]{Financial System Monitoring Department, Central Bank of Brazil, Fortaleza, CE, Brazil \printead[presep={,\ }]{e4}}

\end{aug}

\begin{abstract}
%We present a novel data denoising framework inspired by martingale optimal transport, aimed at identifying the closest distribution to noisy data while enforcing a prescribed structural constraint and self-consistency. The framework is linked to existing denoising approaches, with proofs establishing key properties such as solution existence and robustness. We also introduce Kantorovich dominance, a new distributional dominance relationship, and utilizes it to formulate a modified denoising problem that enhances stability, improves computational efficiency, and yields more meaningful solutions in certain domains. We present numerical examples that illustrate solutions to both the fully self-consistent and Kantorovich dominance problems.

We introduce a new framework for data denoising, partially inspired by martingale optimal transport. For a given noisy distribution (the data),  our approach involves finding the closest distribution to it among all distributions which 1) have a particular prescribed structure (expressed by requiring they lie in a particular domain), and 2)  are self-consistent with the data.  We show that this amounts to maximizing the variance among measures in the domain which are dominated in convex order by the data.  For particular choices of the domain, this problem and a relaxed version of it, in which the self-consistency condition is removed, are intimately related to various classical approaches to denoising.  We prove that our general problem has certain desirable features: solutions exist under mild assumptions, have certain robustness properties, and, for very simple domains, coincide with solutions to the relaxed problem.  

We also introduce a novel relationship between distributions, termed Kantorovich dominance, which retains certain aspects of the convex order while being a weaker, more robust, and easier-to-verify condition. Building on this, we propose and analyze a new denoising problem by substituting the convex order in the previously described framework with Kantorovich dominance. We demonstrate that this revised problem shares some characteristics with the full convex order problem but offers enhanced stability, greater computational efficiency, and, in specific domains, more meaningful solutions. Finally, we present simple numerical examples illustrating solutions for both the full convex order problem and the Kantorovich dominance problem.

\end{abstract}



\end{frontmatter}
%%%%%%%%%%%%%%%%%%%%%%%%%%%%%%%%%%%%%%%%%%%%%%
%% Please use \tableofcontents for articles %%
%% with 50 pages and more                   %%
%%%%%%%%%%%%%%%%%%%%%%%%%%%%%%%%%%%%%%%%%%%%%%
%\tableofcontents

%\section{Introduction}

%This template helps you to create a properly formatted \LaTeXe\ manuscript.
%%%%%%%%%%%%%%%%%%%%%%%%%%%%%%%%%%%%%%%%%%%%%%
%% `\ ' is used here because TeX ignores    %%
%% spaces after text commands.              %%
%%%%%%%%%%%%%%%%%%%%%%%%%%%%%%%%%%%%%%%%%%%%%%
%Prepare your paper in the same style as used in this sample .pdf file.
%Try to avoid excessive use of italics and bold face.
%Please do not use any \LaTeXe\ or \TeX\ commands that affect the layout
%or formatting of your document (i.e., commands like \verb|\textheight|,
%\verb|\textwidth|, etc.).

%\linenumbers

\section{Introduction}\label{sect: intro}

Denoising or dimensional reduction is a  central problem in statistics and machine learning \cite{ChenChiFanMa2021Spec}, \cite{kegl2000learning}, consisting of inferring an unknown probability distribution $\mu$ (the true data or signal) from an observed distribution $\nu$ (the noisy data).

Given a structured domain $\cD$ of Borel probability measures to which the signal is assumed to belong, there are at least two distinct approaches to denoising.  The first involves finding the $\mu \in \cD$ which is \emph{closest to the data} $\nu$ in an appropriate sense.  When the distance between $\mu$ and $\nu$ is measured by the Wasserstein metric (\eqref{eqn: Wasserstein distance} below), this corresponds to the \emph{relaxed problem} in our nomenclature here (\eqref{eqn: relaxed problem} below) and, for different choices of $\cD$, encompasses principal components, $k$-means clustering, and the version of principal curves in \cite{kegl2000learning}.

A second approach involves identifying a $\mu \in \cD$ which \textit{lies in the middle of} $\nu$ in a certain sense.  A precise formulation of this notion is known as self-consistency in the statistics literature and was first formulated by Hastie and Stuetzle in their seminal work \cite{HastiStuetzle1989} introducing principal curves. While self-consistency (expressed in our work here as the existence of a martingale coupling $\pi \in \cM(\mu,\nu)$ defined in \eqref{eqn: martingale couplings} below) can naturally be interpreted as, conditional on the signal being $X$, the average of the noise around $X$ vanishing, it also arises (roughly speaking) as the first order variation of the distance function from the data \cite{HastiStuetzle1989}.  Despite this connection, solutions obtained from the self-consistency approach do not even locally generally minimize the distance to the data among $\mu \in \cD$ \cite{DuchampStuetzle1996}, and, to the best of our knowledge, there is not an existing general paradigm for data denoising capturing key features from both of these two approaches.

Our first contribution here is to propose such a framework, by exploiting ideas from the theory of martingale optimal transport (MOT) \cite{Beiglbock2016}, \cite{henry2017model}, \cite{de2019irreducible}. Heuristically, this problem amounts to finding the closest $\mu$ to $\nu$ \emph{among those $\mu \in \cD$ which can be coupled to $\nu$ in a self-consistent way}; more precisely, it is to minimize
\be\label{introproblem}
\min_{\mu \in \cD} \min_{\pi \in \cM(\mu,\nu)}  \int |x-y|^2d\pi(x,y).
\ee
We note that this can be interpreted as a backwards martingale optimal transport problem somewhat reminiscent of the one in \cite{kramkov2022optimal}.  Since by Strassen's theorem, the measures $\mu$ for which a self-consistent (or martingale, in the language of MOT) coupling to $\nu$ exist are precisely those which are less than $\nu$ in convex order, denoted by $\mu \C \nu$, and the distance between measures is measured by minimizing the expected squared distance among such couplings, it is straightforward to show that this problem is in fact equivalent to \emph{maximizing} the variance among measures $\mu \in \cD$ which are dominated by $\nu$ in convex order.  
We show that this novel denoising problem has certain desirable properties; solutions always exist for reasonably nice domains $\cD$, and are robust in the sense that if the measure $\nu$ is close to some $\rho \in \cD$ (i.e., the noise is small) with $\rho \C \nu$, our solution $\mu$ must be quantifiably close to $\rho$ as well.  Furthermore, for particularly simple domains, we show that problem \eqref{introproblem} in fact coincides with the following \emph{relaxed problem}
\be\label{eqn: intro relaxed problem}
    \min_{\mu \in \cD}  \min_{\pi \in \Pi(\mu,\nu)} \int |x-y|^2d\pi(x,y)
\ee
where $\Pi(\mu,\nu)$ is the set of couplings between $\mu$ and $\nu$ (without the martingale or self-consistency condition imposed).
 This is not surprising, since the self-consistency condition arises as a sort of optimality condition in \eqref{eqn: relaxed problem}; this equivalence is in fact implicit in standard analysis of $k$-means clustering problems (although we have not seen \eqref{introproblem} explicitly formulated in this setting, and we establish the equivalence more generally here by identifying a general condition on the domain $\cD$ for which it holds).

Since they serve as a key motivation for our work here, let us digress briefly to describe in more detail how various notions of principal curves appearing in the literature relate to our framework. Hastie and Stuetzle defined a principal curve of $\nu$ to be a smooth curve $s \in \R \mapsto f(s)$ such that for each $s$, the barycenter of the points $y$ which project to $f(s)$ is $f(s)$.  In our language, letting $P$ be the projection map $P(y) =\argmin_{f(s)} \|y - f(s)\| $, $\mu = P_\# \nu$ and $\pi =(P,\id)_\#\nu \in  \cM(\mu,\nu)$, where $\id$ denotes the identity map and $P_\# \nu$ the pushforward of the data distribution $\nu$ by  $P$.\footnote{Let $F: {\cal X} \to {\cal Y}$ be a measurable map and $\rho$ be a distribution on $\cal X$. The pushforward of $\rho$ by the map $F$, denoted by $F_\# \rho$, is the distribution on $\cal Y$ satisfying $F_\# \rho (B) = \rho ( F^{-1} (B))$ for every measurable set $B$ in $\cal Y$.} Several variants of this definition have been defined since.  Notably, Tibshirani relaxed the projection requirement by looking for (in our nomenclature) a probability measure $\mu$ supported on a smooth curve $s\mapsto f(s)$, together with a martingale coupling $\pi \in \cM(\mu,\nu)$ (so $\mu$ need not be the projection of $\nu$ onto $f$, but must still be in convex order with it) \cite{tibshirani1992principal}.  Neither the definition in \cite{HastiStuetzle1989} nor the one in \cite{tibshirani1992principal} had a variational aspect analogous to \eqref{introproblem}, although the idea of minimizing the distance to the data was clearly present in the formulation of the self-consistency condition as discussed above.  Heuristically, when $\mathcal{D}$ is taken to be the set of all curves (neglecting for now issues about the regularity of the curves), our problem \eqref{introproblem} is to find the principal curve in the sense of Tibshirani which is closest to the data.

Existence of principal curves as defined by Hastie-Stuetzle  for general distributions $\nu$ is not known. Partially to address this, another notion was introduced by Kegl et al \cite{kegl2000learning}.  In our language, their principal curves are solutions to the relaxed problem \eqref{eqn: intro relaxed problem} when $\mathcal{D}$  is the set of all continuous curves of length at most $L$.  With this definition, they easily established existence for any $\nu$.  On the other hand, the self-consistency condition is lost.

For this same domain $\mathcal{D}$, we can define a principal curve as a solution to \eqref{introproblem}; a straightforward argument (given in a more general setting in the first part of Theorem \ref{existence}) then implies existence of a solution.  To the best of our knowledge, this is the first notion of a self-consistent principal curve for which solutions generally exist and have the desirable feature of minimizing noise, or being as close as possible to the data.

Returning to the discussion of general domains $\cD$, despite its advantages outlined above, our formulation \eqref{introproblem} does come with certain drawbacks.  First, checking whether a self-consistent coupling between $\mu \in \cD$ and $\nu$ exists, or, equivalently, checking whether $\mu$ and $\nu$ are in convex order, is not straightforward, and so \eqref{introproblem} is computationally challenging.  Secondly, we demonstrate that \eqref{introproblem} can be \emph{unstable} with respect to variations in the data $\nu$, and third, for certain problems of interest, the domain $\cD$ may contain very few measures $\mu$ such that $\mu \C \nu$, making the problem \eqref{introproblem} trivial and its solution uninformative (see the example in Section \ref{applications} and in particular Remark \ref{rem: denoising with convex order is trivial}).  

To address these issues, we introduce a weakening of the convex order relation, and a corresponding variational problem (see \eqref{weakproblem} below). The new dominance relation between two probabiltiy measures,  which we call the Kantorovich dominance relation (KDR in short), amounts to imposing the existence of a coupling between $\mu$ and $\nu$ which, though not necessarily self-consistent, enjoys some features of self-consistency; the resulting denoising problem \eqref{weakproblem} therefore falls \emph{in between} the original problem \eqref{introproblem} with the full self-consistency condition and the fully relaxed problem \eqref{eqn: intro relaxed problem}.  We argue that this dominance relation is natural, by demonstrating that, like self-consistency, it also arises as an optimality condition for the relaxed problem in a certain sense; indeed, for a large class of domains, which we name cones, we show that the new problem is equivalent to the relaxed problem.  In addition, we show that, like \eqref{introproblem}, \eqref{weakproblem} enjoys quantifiable robustness properties as the noise becomes small, but, in contrast to \eqref{introproblem}, \eqref{weakproblem} is stable with respect to perturbations in the data distribution $\nu$.

We illustrate the properties of our new order dominance relation by discussing its relationship to several established data analysis techniques, including PCA \cite{Jolliffe2002}, \cite{greenacre2022principal}, $k$-means clustering \cite{lloyd1982least}, \cite{kanungo2002efficient}, principal curves \cite{lee:etal:20}, \cite{kang:oh:24}, and Gaussian denoising \cite{ChenChiFanMa2021Spec}; (versions of) each of these arise for appropriate choices of the domain $\cD$.

The remainder of this article is organized as follows. 
In Section \ref{sec2}, we present the full, self-consistency problem formulation, establish the existence of an optimal solution under the convex ordering constraint, and provide examples illustrating a variety of feasible domains.
In Section \ref{sec3}, we introduce the Kantorovich dominance and investigate its key properties.
In Section \ref{sec4}, we examine the variance maximization problem under the Kantorovich dominance, along with equivalent optimization formulations and applications. Section \ref{sec_numerics} is reserved for numerical examples.

\section{General problem formulation and basic properties}\label{sec2}
Let $\cP_2(\R^d)$ denote the set of all probability measures on $\R^d$ with finite second moments, and let $\cP_{2,0}(\R^d) = \{ \mu \in \cP_2(\R^d) \mid \int x \, d\mu(x) = 0\}$ represent the set of centered probability measures in $\cP_2(\R^d)$.  In this paper, all probability measures are assumed to be in $\cP_2(\R^d)$, unless stated otherwise.

For $\mu, \nu \in \cP_2(\R^d)$, the Wasserstein distance ${\mathcal W}_2(\mu, \nu)$ between $\mu$ and $\nu$ is defined as
\begin{equation}\label{eqn: Wasserstein distance}
\Wtq(\mu, \nu) = \inf_{\pi \in \Pi(\mu, \nu)} \int_{\R^d \times \R^d} |x - y|^2 \, d\pi(x, y),
\end{equation}
where $\Wtq(\mu, \nu) = \big({\mathcal W}_2(\mu, \nu) \big)^2$, and $\Pi(\mu, \nu)$ is the set of all couplings of $\mu$ and $\nu$, i.e.,
\[
\Pi(\mu, \nu) = \{ \pi = \cL(X,Y) \mid \cL(X) = \mu,\, \cL(Y) = \nu \},
\]
where $\cL(X)$ denotes the law (distribution) of the random variable $X$ \cite{santambrogio2015optimal}, \cite{villani2003topics}. We denote $\cL(X) = \mu$ and $ X \sim \mu$ interchangeably. $| \cdot |$ denotes the Euclidean norm.

In a typical application, we will assume $Y \sim \nu \in \cP_{2,0}(\R^d)$ is the given data distribution, satisfying the following model assumption:
\begin{equation} \label{eqn: model}
Y = X + R,
\end{equation}
where $X \sim \rho \in \cP_{2,0}(\R^d)$, $R \sim \epsilon \in \cP_{2,0}(\R^d)$. Our goal is to recover \( \rho \) from the observed data \( \nu \) perturbed by \( \epsilon \). With prior knowledge or constraint about \( \rho \), we assume it belongs to a known domain \( \mathcal{D} \subset \mathcal{P}_{2,0}(\mathbb{R}^d) \). Given \( \nu \in \mathcal{P}_{2,0}(\mathbb{R}^d) \), we thus arrive at the following problem (which is simply \eqref{introproblem} rewritten in probabilistic notation):
\be\label{problem}
\min_{\mu \in \cD} \min_{\pi \in \cM(\mu,\nu)} \E_\pi[ | X - Y |^2],
\ee
where $\cM(\mu,\nu)$ represents the set of martingale couplings between $\mu$ and $\nu$:
\begin{equation}\label{eqn: martingale couplings}
\cM(\mu,\nu) = \{ \pi = \cL(X, Y) \mid \cL(X) = \mu,\, \cL(Y) = \nu,\, \E_\pi[Y|X] = X \}.
\end{equation}
The martingale condition is often referred to as  \emph{self-consistency} in the statistics literature.  In the generic model \eqref{eqn: model}, it says that conditional on $X$, the average noise vanishes, $\E[R|X] = 0$.

We say that $\mu$ and $\nu$ in $\cP_2(\R^d)$  are in convex order, denoted by $\mu \C \nu$, if $\int f \, d\mu \leq \int f \, d\nu$ for any convex function $f$. Strassen’s theorem states that $\cM(\mu,\nu) \neq \emptyset$ if and only if $\mu \C \nu$.

If $\pi \in \cM(\mu,\nu)$ is a martingale coupling of $X \sim \mu$ and $Y \sim \nu$, we have $\E_\pi[XY] = \E_\mu[X \E_\pi[Y|X]] = \E_\mu[|X|^2]$. Therefore, if $\mu, \nu \in \cP_{2,0}(\R^d)$, it follows that
\[
    \E_\pi[|X - Y|^2] = \E_\nu[|Y|^2] - \E_\mu[|X|^2] = \Var(\nu) - \Var(\mu),
\]
where $\Var(\mu)$ denotes the variance of $\mu$, defined as $
    \Var(\mu) = \int |x - \int x \, d\mu(x)|^2  d\mu(x)$. 
    
    Since $\nu$ is given and fixed, this shows that problem \eqref{problem} can equivalently be formulated as:
\be\label{maxvar}
    \max_{\mu \in \cD,\, \mu \C \nu} \Var(\mu).
\ee

\begin{remark}
Since $\cD$ typically consists of measures supported on low dimensional spaces (for instance, curves), this formulation appears to achieve a spectacular dimensional reduction, as it involves optimizing over $\mu \in \cD$, rather than $(\mu,\pi)$ with $\mu \in \cD$ and $\pi \in \cM(\mu,\nu)$.  The catch, of course, is that the constraint $\mu \C \nu$ involves $\nu$ in a sophisticated way and is not straightforward to check.
\end{remark}

\begin{remark}\label{Lipconvexorder}
We have $ \mu \C \nu$ if and only if
$\int f\,d\mu \le \int f\,d\nu$ for every Lipschitz convex function $f$.
To see this, assume the inequality holds for all Lipschitz convex $f$, and let $g$ be any convex lower semicontinuous function. Then there exists a sequence of Lipschitz convex functions (indeed, one may take piecewise affine convex functions) $(g_n)_{n\ge1}$ which increases to $g$ pointwise. For each $n$,
$\int g_n\,d\mu \le \int g_n\,d\nu \le \int g\,d\nu$. Letting $n\to\infty$ and applying the monotone convergence theorem yields
$\int g\,d\mu \le \int g\,d\nu$,
which is the convex order condition.
\end{remark}

Note that, with the definition of Wasserstein distance in hand, we can rewrite the relaxed problem \eqref{eqn: intro relaxed problem} from the introduction as:
\be\label{eqn: relaxed problem}
    \min_{\mu \in \cD} \mathcal{W}_2^2(\mu, \nu).
\ee

Since self-consistency can be seen as a first order optimality condition for the functional in \eqref{eqn: relaxed problem} \cite{HastiStuetzle1989}, it is not surprising that for certain simple domains $\mathcal{D}$, problems \eqref{problem} and \eqref{eqn: relaxed problem} are in fact equivalent.  This is well known for problems such as $k$-means clustering (see Example \ref{ex: discrete} below); we offer here a general condition on $\cD$ under which equivalence holds.

The $\pi$-conditional barycentric recentering operation described below is a well-known procedure in the $k$-means clustering algorithm. (For this connection, one may assume $\mu$ is supported on $k$ points in $\R^d$.)

\begin{definition}\label{def: centering map}
Given $\mu, \nu \in \cP_2(\R^d)$ and a coupling $\pi = \mu \otimes \pi_x \in \Pi(\mu,\nu)$,\footnote{A disintegration of $\pi \in \Pi(\mu,\nu)$ w.r.t. $\mu$, denoted by $\pi = \mu \otimes \pi_x$, means that for any Borel sets $A, B \subset \R^d$, it holds $\pi(A \times B) = \int_A \pi_x(B) d\mu(x)$. Using conditional probability notation, $\pi_x (B) = \cP( Y \in B \, | \, X = x)$.} define the $\pi$-conditional barycentric map $c_\pi(x) := \int y\, d\pi_x(y)$ for $\mu$-a.e. $x$. We call ${c_\pi}_\#(\mu)$ the $\pi$-conditional barycentrically recentered first marginal of $\pi$, and  $\mathsf{C}_\# \pi$ as the $\pi$-conditional barycentrically recentered coupling of $\pi$, where $\mathsf{C}(x,y) = \mathsf{C}_\pi (x,y) := (c_\pi(x), y)$.
\end{definition}

Note that $\mathsf{C}_\# \pi \in \cM({c_\pi}_\#(\mu), \nu )$ is a martingale measure, and thus ${c_\pi}_\#(\mu) \C \nu$.

\begin{proposition}\label{prop: equivalence between relaxed and full problems}
Assume that there exists an optimal $\mu \in \cD$ for the relaxed problem \eqref{eqn: relaxed problem} and an optimal $\pi \in \Pi(\mu,\nu)$ for  \eqref{eqn: Wasserstein distance} such that ${c_\pi}_\#(\mu) \in \cD$.  Then ${c_\pi}_\#(\mu)$ is optimal in both \eqref{eqn: relaxed problem} and \eqref{problem}, and consequently, the two problems are equivalent (by yielding the same value).
\end{proposition}
\begin{proof}
    Since $c_\pi(x)$ is the barycenter of $\pi_x$, we must have
    \begin{eqnarray*}
    \int_{\mathbb{R}^d \times \mathbb{R}^d}|x-y|^2 d\mathsf{C}_\# \pi(x,y) &=&\int_{\mathbb{R}^d \times \mathbb{R}^d}|c_\pi(x)-y|^2d \pi(x,y)\\  &=&\int_{\mathbb{R}^d}\left(\int_{\mathbb{R}^d}|c_\pi(x)-y|^2d \pi_x(y)\right)d\mu(x)\\ &\leq &\int_{\mathbb{R}^d}\left(\int_{\mathbb{R}^d}|x-y|^2d \pi_x(y)\right)d\mu(x) 
    =\mathcal{W}_2^2(\mu, \nu),
     \end{eqnarray*}
     where the last equality is from the assumption that $\pi$ is optimal for \eqref{eqn: Wasserstein distance}. Since $\mathcal{W}_2^2(\mu, \nu)$ is equal to the minimal value in \eqref{eqn: relaxed problem} (since $\mu \in \cD$) and less than or equal to the minimal value in \eqref{problem} (since ${c_\pi}_\#(\mu) \in \cD$), this implies optimality of ${c_\pi}_\#(\mu)$ in both problems.
\end{proof}

We now turn our attention to the existence and robustness of solutions.  Given $\nu \in \cP_2(\R^d)$, we define the set of probability measures less than or equal to $\nu$ in convex order as:
\be
    \cM_{\nu} := \{ \mu \in \cP_2(\R^d) \mid \mu \C \nu \}.
\ee
\begin{lemma}\label{compactness}
   %\leavevmode\newline
    \begin{enumerate}[label=\roman*), ref=\roman*]
    \item  $ \cM_\nu$ is compact in the $ \Wt $-metric for any $\nu \in \cP_2(\R^d)$.

        \item Let $ \Wt(\nu_n, \nu) \to 0 $ as $ n \to \infty $, and $ \mu_n \in \cM_{\nu_n} $. Then the sequence $ \{ \mu_n \}_n $ is precompact, meaning there exists a subsequence $ \{ \mu_{n_k} \}_k $ that converges in $ \Wt $ to some $ \mu \in \cM_\nu $.
    \end{enumerate}
\end{lemma}

\begin{proof}

i) The set $\cM_\nu$ is clearly closed: If $\{\mu_n\}_n \subset \cM_\nu$ and $\lim_{n \to \infty} \Wt(\mu_n, \mu) = 0$ so that $\int f \, d\mu_n \le \int f \, d\nu$ for any $L$-Lipschitz convex function $f$, then by duality, we have
\[
\left | \int f \, d\mu_n - \int f \, d\mu \right | \le L {\cal W}_1(\mu_n, \mu) \le L {\cal W}_2(\mu_n, \mu) ,
\]
proving $\int f \, d\mu \le \int f \, d\nu$. Hence, $\mu \in \cM_\nu$ (see Remark \ref{Lipconvexorder}). 

To show $\cM_\nu$ is compact, let $\{\mu_n\}_n \subset \cM_\nu$. By monotone convergence, for $\epsilon > 0$, there exists $a > 0$ with $\int (2|x|^2 - a^2)^+ \, d\nu < \epsilon$. Then we have
\begin{align*}
    \int (2|x|^2 - a^2)^+ \, d\nu &\geq \int (2|x|^2 - a^2)^+ \, d\mu_n \geq \int_{|x| \ge a} (2|x|^2 - a^2)^+ \, d\mu_n \geq \int_{|x| \ge a} |x|^2 \, d\mu_n.
\end{align*}
By the second inequality, $\{\mu_n\}_n$ is tight. Hence there is a subsequence $\{\mu_k\}_k$ and a probability measure $\mu$ such that $\mu_k \to \mu$ weakly. The lemma will follow if $\Wt(\mu_k, \mu) \to 0$. By \cite[Theorem 7.12]{villani2003topics}, it suffices to show that for any $\epsilon > 0$, there is $R > 0$ such that $\int_{|x| \ge R} |x|^2 \, d\mu_k(x) < \epsilon$ for all $k$. This is shown by the third inequality, completing the proof.

ii) As before, for $\epsilon > 0$, there exists $a > 0$ such that $\int (|x| - a)^+ \, d\nu < \epsilon$. Then,
\[
    \mu_n(\R^d \setminus B_{a+1}(0)) \leq \int (|x| - a)^+ \, d\mu_n \leq \int (|x| - a)^+ \, d\nu_n < \epsilon,
\]
for all large $n$, since $\Wt(\nu_n, \nu) \to 0$ implies $\int (|x| - a)^+ \, d\nu_n \to \int (|x| - a)^+ \, d\nu$. This yields a subsequence $\{\mu_k\}_k$ and a probability measure $\mu$ such that $\mu_k \to \mu$ weakly. As before, there exists $a > 0$ such that $\int (2|x|^2 - a^2)^+ \, d\nu < \epsilon$, implying
\[
    \int_{|x| \ge a} |x|^2 \, d\mu_k \leq \int (2|x|^2 - a^2)^+ \, d\nu_k < \epsilon,
\]
for all large $k$. By \cite[Theorem 7.12]{villani2003topics}, we conclude that $\Wt(\mu_k, \mu) \to 0$.
\end{proof}

\begin{remark}
The preceding lemma holds with an almost identical proof hold if we replace $\Wt$ with  the $p$-Wasserstein distance,  defined by $\mathcal{W}_p^p(\mu, \nu) := \inf_{\pi \in \Pi(\mu, \nu)} \int_{\R^d \times \R^d} |x - y|^p \, d\pi(x, y)$ for $ p \in [1, \infty) $, asserting that $\cM_{\nu, p} := \{ \mu \in \cP_p(\R^d) \mid \mu \C \nu \},$ is $\mathcal{W}_p$-compact.
\end{remark}

We now show that our optimization problem \eqref{maxvar} admits a solution and can recover the true distribution $\rho$ as the noise diminishes.
\begin{theorem}\label{existence}
   % \leavevmode\newline
    \begin{enumerate}[label=\roman*), ref=\roman*]
        \item If $\cD \subset \cP_2(\R^d)$ is closed under the $\Wt$-metric and  $\cD \cap \cM_\nu$ is non-empty, then  problem \eqref{maxvar} attains a solution.
        \item Let $\mu^*$ be a solution to \eqref{maxvar}. Then for any $\rho \in \cD$ with $\rho \C \nu$, we have:
        \[
            \Wt(\mu^*, \rho) \leq \sqrt{\Var(\nu) - \Var(\rho)} + \Wt(\nu, \rho).
        \]
        Consequently, $\Wt(\mu^*, \rho) \to 0$ as $\Wt(\nu, \rho) \to 0$, i.e., as the noise diminishes.
    \end{enumerate}
\end{theorem}

\begin{proof}
    i) By Lemma \ref{compactness} the set $\cM_\nu$ is $\Wt$-compact. Hence, $\cD \cap \cM_\nu$ is $\Wt$-compact. Since the functional $\mu \mapsto \Var(\mu)$ is continuous in the $\Wt$-metric, part i) follows.

    ii) Recall that $\displaystyle \Wt(\mu, \nu) = \min_{\pi \in \Pi(\mu, \nu)} \sqrt{\E_\pi[|X - Y|^2]}$. We proceed as follows:
    \begin{align*}
        \Wt(\mu^*, \rho) & \leq \Wt(\mu^*, \nu) + \Wt(\nu, \rho) \\
        & \leq \sqrt{\E_\pi[|X - Y|^2]} + \Wt(\nu, \rho), \quad \text{for any} \, \pi \in \cM(\mu^*, \nu) \\
        & = \sqrt{\Var(\nu) - \Var(\mu^*)} + \Wt(\nu, \rho) \\
        & \leq \sqrt{\Var(\nu) - \Var(\rho)} + \Wt(\nu, \rho),
    \end{align*}
    where the last inequality follows from the optimality of $\mu^*$ in problem \eqref{maxvar}.
\end{proof}
\subsection{Examples of domains}
The following examples clarify the connection between our framework, $k$-means clustering, and principal curves.

\begin{example}[$\cD$ as the set of measures on Lipschitz curves]\label{ex: lipschitz}
Let $\Omega \subset \R^d$ be a compact and convex set, and fix parameters \( L, T > 0 \). Consider the set of Lipschitz curves and the probability measures supported on them, defined by
\begin{align}
\cC_L &= \big\{ \alpha : [0, T] \to \Omega \,\big|\, | \alpha(t) - \alpha(s) | \le L |t-s| \big\},\\
\cD_L &= \big\{ \mu \in \cP_2(\R^d) \,\big|\, \spt(\mu) \subset {\rm Im}(\alpha) \text{ for some } \alpha \in \cC_L \big\},
\end{align}
where ${\rm Im}(\alpha) = \{ \alpha(t) \mid t \in [0,T] \}$ and $\spt(\mu)$ denotes the support of $\mu$. In practice, $\Omega$ can be the convex hull of the support of $\nu$, or a closed ball that contains $\nu$'s support.

To show $\cD_L$ is closed in $\Wt$-metric, let $(\mu_n)_n$ be a sequence in $\cD_L$ converging to $\mu$, and let $(\alpha_n)_n \subset \cC_L$ with $\spt(\mu_n) \subset {\rm Im}(\alpha_n)$. We need to show that $\spt(\mu) \subset {\rm Im}(\alpha)$ for some $\alpha \in \cC_L$. By the Arzelà–Ascoli theorem, there exists a subsequence of $(\alpha_n)_n$ (which we still denote by $(\alpha_n)_n$) that converges uniformly on $[0,T]$ to some $\alpha \in \cC_L$. This implies $\spt(\mu) \subset {\rm Im}(\alpha)$ as follows: for any $\epsilon > 0$, there exists $N$ such that for all $n \ge N$, ${\rm Im}(\alpha_n) \subset \mathcal{N}_\epsilon({\rm Im}(\alpha))$, where
\[
\mathcal{N}_\epsilon({\rm Im}(\alpha)) := \big\{ x \in \R^d \mid |x-y| \le \epsilon \text{ for some } y \in {\rm Im}(\alpha) \big\}.
\]
Since $\spt(\mu_n) \subset {\rm Im}(\alpha_n)$ and $\mu_n \to \mu$, we have $\spt(\mu) \subset \mathcal{N}_\epsilon({\rm Im}(\alpha))$. Taking the limit as $\epsilon \to 0$ yields $\spt(\mu) \subset {\rm Im}(\alpha)$.
We note that the relaxed problem \eqref{eqn: relaxed problem} for this domain is equivalent to the version of principal curves proposed in \cite{kegl2000learning}.
\end{example}

\begin{example}[$\cD$ as the set of measures on monotone increasing curves]\label{ex: monotone}

A set $\Gamma \subset \R^2$ is said to be monotone if, for any $(x_1, x_2), (y_1, y_2) \in \Gamma$, the following condition holds:
\be
    (y_1 - x_1)(y_2 - x_2) \ge 0.
\ee
Let $\text{MON}$ denote the collection of all monotone sets in $\R^2$. We define the search space as
\be
    \cD_{\text{MON}} = \left\{ \mu \in \cP(\R^2) \mid \text{spt}(\mu) \subset \Gamma, \Gamma \in \text{MON} \right\}.
\ee
To show that $\cD_{\text{MON}}$ is closed under the ${\mathcal W}_2$-metric, consider a sequence $\mu_n \to \mu$ in ${\mathcal W}_2$ with $\mu_n \in \cD_{\text{MON}}$. ${\mathcal W}_2$-convergence implies weak convergence, so by the Portmanteau theorem, for any $(x_1, x_2), (y_1, y_2) \in \text{supp}(\mu)$, there exist $(x_1^n, x_2^n), (y_1^n, y_2^n) \in \text{supp}(\mu_n)$ such that $(x_1^n, x_2^n) \to (x_1, x_2)$ and $(y_1^n, y_2^n) \to (y_1, y_2)$. By continuity of the product function, we have $(y_1 - x_1)(y_2 - x_2) \ge 0$ for any $(x_1, x_2), (y_1, y_2) \in \text{spt}(\mu)$, confirming that $\mu \in \cD_{\text{MON}}$.
Monotonicity, reflecting an increasing dependence on the signal variables $X_1,X_2$, is a natural modelling assumption in many situations (such as when the noisy data variables $Y_1,Y_2$ are highly correlated).  It is also closely related to the Lipschitz condition in the preceding example, as monotone sets $\Gamma \subset \mathbb{R}^2$ are well known to be $1$-Lipschitz graphs $(\bar x_1, \bar x_2) =(\bar x_1, F(\bar x_1))$ of the anti-diagonal  $\bar x_2 =\frac{1}{\sqrt{2}}[x_1-x_2]$ over the diagonal $\bar x_1 =\frac{1}{\sqrt{2}}[x_1+x_2]$ \cite{Minty62}.
\end{example}

\begin{remark}
The last example is closely related to the backwards martingale optimal transport problem studied in \cite{kramkov2022optimal}. For a given $\nu \in \cP_2(\mathbb{R}^2)$, they attempt to minimize $\mathbb{E}_\pi[(Y_1 - X_1)(Y_2 - X_2)]$ among measures $\mu \in \cP_2(\mathbb{R}^2)$ and martingale couplings $\pi \in \cM(\mu,\nu)$.  This differs from our framework in that their cost function $c(x, y) = (y_1 - x_1)(y_2 - x_2)=(\bar y_1-\bar x_1)^2 - (\bar y_2 -\bar x_2 )^2$ is increasing along the diagonal direction %$\bar x_1 :=(x_1 +x_2)/2, \bar y_1 :=(y_1 +y_2)/2$ 
but decreasing along the anti-diagonal direction %$\bar x_2 :=(x_1 -x_2)/2, \bar y_2 :=(y_1 -y_2)/2$ 
and they do not restrict to measures $\mu$ in any particular domain $\cD$.  Despite the latter difference, they show that their optimal $\mu$ in fact belongs to $\cD_{\text{MON}}$. Thus, solving either the problem in the last example or the problem in \cite{kramkov2022optimal} yields a measure $\mu \in \cD_{\text{MON}}$ with $\mu \C \nu$. These two measures need not be the same, as shown in Supplementary Material A.
\end{remark}

The following example relates our framework to data clustering problems.

\begin{example}[$\cD$ as the set of measures supported on sets of bounded cardinality] \label{ex: discrete}
Fix $m \in \N$. Consider the following search space
\begin{align}\label{msupport}
\cD^m 
= \bigg\{ \mu =  \sum_{i=1}^m u_i \delta_{x_i} \,\bigg|\, x_i \in \R^d,\, u_i \ge 0 \  \forall i,\, \sum_{i=1}^m u_i = 1 \bigg \}.
\end{align}
We also consider the set of discrete measures of fixed weight as follows. Let $u = (u_1,...,u_m)$ be a vector of fixed weights, where $u_i \geq 0$ and $\sum_{i=1}^m u_i = 1$. We set 
\begin{align}\label{msupportfixed}
\cD^m_u = \bigg\{ \mu  = \sum_{i=1}^m u_i \delta_{x_i} \,\bigg|\, x_i \in \R^d \bigg\}.
\end{align}
Proposition \ref{prop: equivalence between relaxed and full problems} implies that \eqref{problem} and \eqref{eqn: relaxed problem} are equivalent for either of these domains.  Note that the relaxed problem \eqref{eqn: relaxed problem} with $\cD^m$ is exactly the classical $k$-means clustering problem, while with $\cD^m_u$ it is a fixed weight clustering problem.

\end{example}

\subsection{Stability}
The following example illustrates that solution stability may not hold if the input data $\nu$ varies. Note that this differs from the scenario of diminishing noise.

\begin{example}[Instability of denoising with self-consistency]\label{counterexample:stability}
    We demonstrate that the stability of problem \eqref{problem} does not generally hold, inspired by \cite{bruckerhoff2022instability}.
    
    Define the one-step probability kernel $\kappa_\theta$ from $\R^2$ to $\cP(\R^2)$ by
    \be
        \kappa_\theta(x_1, x_2) = \tfrac{1}{2} \delta_{(x_1 + \cos(\theta), x_2 + \sin(\theta))} + \tfrac{1}{2} \delta_{(x_1 - \cos(\theta), x_2 - \sin(\theta))}.
    \ee
 Define $\mu = \tfrac{1}{3} \delta_{(-1,0)} + \tfrac{1}{3} \delta_{(0,0)} + \tfrac{1}{3} \delta_{(1,0)}$, $\theta_n = \pi (1 - \frac{1}{2(n+1)})$,
 and $\nu_n$ as the second marginal of the martingale transport $\mu \otimes {\kappa_{\theta_n}} \in \cM(\mu, \nu_n)$, depicted as red dots in Figure \ref{fig:counterexample-1}) for each $n \in \mathbb{N}$. Then as $n \to \infty$, we have $\nu_n \to \nu_\infty := \kappa_\pi \mu = \frac{1}{6} \delta_{(-2,0)} + \frac{1}{6} \delta_{(-1,0)} + \frac{1}{3} \delta_{(0,0)} + \frac{1}{6} \delta_{(1,0)} + \frac{1}{6} \delta_{(2,0)}$.
    \begin{figure}[ht]
        \centering
        \begin{minipage}{0.4\textwidth}
            \centering
            \begin{tikzpicture}
                \begin{axis}[
                    axis lines = middle,
                    xlabel = {$x$},
                    ylabel = {$y$},
                    xmin=-3, xmax=3,
                    ymin=-1, ymax=1,
                    width=\linewidth,
                    height=3.5cm,
                    legend pos=north east,
                    legend cell align={left},
                    legend style={font=\tiny},
                    axis equal
                ]
                % Plotting the measure mu
                \addplot[only marks, mark=square*, mark size=2pt, blue] coordinates {(-1,0) (0,0) (1,0)};
                \addlegendentry{$\mu$}
                
                % Plotting the measure nu_{3pi/4}
                \addplot[only marks, mark=*, mark size=2pt, red] coordinates {
                    (-1 - 0.707, 0.707)
                    (-1 + 0.707, -0.707)
                    ( - 0.707, 0.707)
                    ( + 0.707, -0.707)
                    (1 - 0.707, 0.707)
                    (1 + 0.707, -0.707)
                };
                \addlegendentry{$\nu_{1}$}
                \end{axis}
            \end{tikzpicture}
            \caption{Support of $\mu$ and $\nu_{1}$}
            \label{fig:counterexample-1}
        \end{minipage}%
        \begin{minipage}{0.4\textwidth}
            \centering
            \begin{tikzpicture}
                \begin{axis}[
                    axis lines = middle,
                    xlabel = {$x$},
                    ylabel = {$y$},
                    xmin=-3, xmax=3,
                    ymin=-1, ymax=1,
                    width=\linewidth,
                    height=3.5cm,
                    legend pos=north east,
                    legend cell align={left},
                    legend style={font=\tiny},
                    axis equal
                ]
                % Plotting the measure mu
                \addplot[only marks, mark=square*, mark size=2pt, blue] coordinates {
                    (-2, 0)
                    (-1, 0)
                    (0, 0)
                    (1, 0)
                    (2, 0)
                };
                \addlegendentry{$\mu_\infty$}
                
                % Plotting the measure nu_{pi}
                \addplot[only marks, mark=*, mark size=2pt, red] coordinates {
                    (-2, 0)
                    (-1, 0)
                    (0, 0)
                    (1, 0)
                    (2, 0)
                };
                \addlegendentry{$\nu_\infty$}
                \end{axis}
            \end{tikzpicture}
            \caption{Support of $\mu_\infty = \nu_{\infty}$}
            \label{fig:counterexample-2}
        \end{minipage}
    \end{figure}
    
Set $\cM_{\nu,\text{MON}}  = \{\mu \in \cD_{\text{MON}} \mid \mu \C \nu\}$. 
    For each $n \in \N \cup \{\infty\}$, let $\mu_n$ be the solution to \eqref{maxvar} with respect to $\nu_n$, which aims to maximize the variance by optimally spreading out the mass. 
    This results in $\mu_n = \mu$ for all $n \in \N$ due to the monotonicity constraint. However, in the limit $n \to \infty$, $\nu_\infty$ itself is monotone, implying $\mu_\infty = \nu_\infty$. This example shows that $\mu$ is not stable as a solution to \eqref{maxvar} with respect to $\nu$ in the sense that $\nu_n \to \nu_\infty \centernot\implies \mu_n \to \mu_\infty$, demonstrating the instability of problem \eqref{problem}.
\end{example}

In contrast to \eqref{problem}, the relaxed problem \eqref{eqn: relaxed problem} enjoys the following stability property.
\begin{proposition}[Stability of the relaxed problem]
If $\nu_n \to \nu$ in $\Wt$ metric, and $\mu_n \in \argmin_{\mu \in \cD} \mathcal{W}_2^2(\mu, \nu_n)$ converges in $\Wt$ to $\mu^*$, then $\mu^* \in \argmin_{\mu \in \cD} \mathcal{W}_2^2(\mu, \nu)$.
\end{proposition}

This proposition implies that if $\cD$ is compact and \eqref{eqn: relaxed problem} has a unique minimizer $\mu^*$ for $\nu$, then any sequence of minimizers $\mu_n$ for $\nu_n$ converges to $\mu^*$.

\begin{proof}
Since $\mu_n$ minimizes $\mathcal{W}_2^2(\cdot, \nu_n)$ over $\cD$, we have
\[
    \mathcal{W}_2^2(\mu_n, \nu_n) \leq \mathcal{W}_2^2(\xi, \nu_n) \quad \forall \xi \in \cD.
\]
By the continuity of $\Wt$ with respect to its own topology, taking the limit yields
\[
    \Wtq(\mu^*, \nu) = \lim_{n \to \infty} \Wtq(\mu_n, \nu_n) \leq \lim_{n \to \infty} \mathcal{W}_2^2(\xi, \nu_n) = \Wtq(\xi, \nu) \quad \forall \xi \in \cD.
\]
This confirms that $\mu^*$ is a minimizer of $\mathcal{W}_2^2(\cdot, \nu)$ over $\cD$.
\end{proof}

\subsection{Preliminary numerical simulations}
As a proof of concept, we generate 110 and 1000 observations of $Y = X + R$, where $X$ is monotone and $R$ is Gaussian noise. We set the domain $\cD$ to be either the set of measures supported on a monotone set or the set of measures supported on a curve with finite length. We solve problem \eqref{problem} by adapting a generalized Lloyd algorithm.   The self-consistency constraint is softly enforced via a penalty term in the objective function. For comparison, we also generated principal curves based on different definitions from \cite{HastiStuetzle1989}, \cite{kegl2000learning}, and \cite{tibshirani1992principal}.\footnote{For the method proposed by \cite{tibshirani1992principal}, we only generate the curve using 110 data points, as we encounter overflow issues with larger sample sizes.} The results are shown in Figure \ref{fig:monotone}. The figures illustrate that the curve fitting method based on problem \eqref{problem} performs reasonably well in capturing the underlying structure of the data compared to other principal curve methods.
\begin{figure}[ht]
    \centering
\includegraphics[width=0.95\textwidth, trim=10 2 10 3.5, clip]{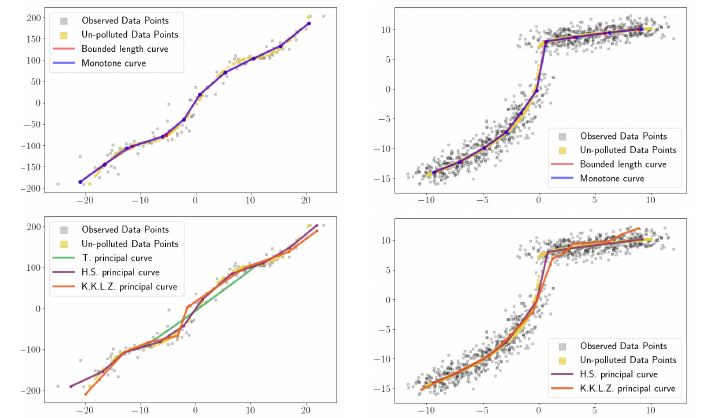}
    \caption{
Principal curve fitting methods applied to 110 data points (Left) and 1000 monotone data points (Right). "Monotone curve" and "Bounded length curve" correspond to our solutions of problem \eqref{problem} over the domain $\cD$, defined as either the set of measures supported on a monotone set or those supported on a curve with finite length. These solutions are computed using a generalized Lloyd algorithm, where the self-consistency constraint is softly enforced via a penalty term in the objective function. 
For comparison, we also present the results of fitting principal curves using the methods by Hastie-Stuetzle \cite{HastiStuetzle1989}, Kégl-Krzyzak-Linder-Zeger \cite{kegl2000learning}, and Tibshirani \cite{tibshirani1992principal}.
}
    \label{fig:monotone}
\end{figure}

\section{The Kantorovich dominance relation}\label{sec3}

Despite its conceptual appeal, the learning problem with self consistency \eqref{problem} (or equivalently \eqref{maxvar}) has certain drawbacks: the convex order constraint is difficult to check, the problem may be unstable (recall Example \ref{counterexample:stability}), and the intersection of the certain natural domains $\cD$ with  the set $\{\mu \, | \, \mu \C \nu\}$ may be nearly empty, and thus not contain a reasonable solution (see, for example, the application in Section \ref{applications}).

These challenges prompt us to introduce a weaker relation, the Kantorovich dominance relation, which still captures some aspects of the martingale property, while avoiding some of the computational and theoretical difficulties associated with the convex order relation. 

Let $\langle \cdot , \cdot \rangle$ denote the inner product. For $\mu, \nu \in \cP_2(\R^d)$, we say that $\mu$ is less than $\nu$ in the Kantorovich dominance relation (KDR), and write $\mu \K \nu$, if there is $\pi \in \Pi(\mu,\nu)$ such that
\be\label{weakconvexorder}
 \int \langle x - b_\nu, y-x \rangle d\pi(x,y) = 0, \  \text{ where } \ b_\nu = \int y \, d\nu(y) \text{ is the barycenter of } \nu.
\ee
Let $T_k(x) = x + k$ be a translation map. \eqref{weakconvexorder} says \(\mu \K \nu\) iff the translated measures \( \mu' = {T_{-b_\nu}}_\#(\mu)\) and \( \nu' = {T_{-b_\nu}}_\# (\nu) \) (thus $\nu'$ is centered) admits a coupling \(\pi' \in \Pi(\mu',\nu')\) such that
\[
\int \langle x, y-x \rangle \,d\pi'(x,y)=0.
\]
By definition, the KDR relation is translation invariant: if \(\mu \K \nu\), then for any common shift vector \(a\), the translated measures \(\mu_a := (T_a)_\#\mu\) and \(\nu_a := (T_a)_\#\nu\) also satisfy \(\mu_a \K \nu_a\). Accordingly, throughout this paper we assume for simplicity that \(\nu\) is centered (i.e., \(\nu \in \cP_{2,0}(\R^d)\)), unless stated otherwise. Since \(\nu\) represents the given data distribution, this centering step is straightforward in practice.

\begin{remark}
    If $\mu \C \nu$, then there exists a martingale coupling $\pi  \in \cM(\mu,\nu)$ such that
    \[
    \int \langle x, y-x \rangle d\pi(x,y) = \int \langle x, \Big(\int (y-x) d\pi_x(y) \Big) \rangle d\mu(x) = 0.
    \]
Therefore, the KDR is a relaxation of the convex order relation, capturing a weaker form of the martingale property.  Precisely,  $v(x): =\int (y-x) d\pi_x(y) $ for some coupling $\pi$, the KDR requires only that $v$ is orthogonal to $\id$ in $L^2(\mu)$.
\end{remark}

A key part of the motivation for the Kantorovich dominance is that it arises as a sort of optimality condition for the relaxed problem \eqref{eqn: relaxed problem} for a certain class of domains (see Theorem \ref{thm: kant equivalence to relaxed prob} below), analagously to the  way that self consistency arises as an optimality condition for particularly simple domains (Proposition \ref{prop: equivalence between relaxed and full problems}).

\begin{remark}
If $\nu$ is centered, for $\pi = \mu \otimes \nu$ (the product measure), we have
   \[
   \int \langle x, y-x \rangle d\pi(x,y) = \langle \int x\, d\mu(x), \int y \, d\nu(y) \rangle  - \int |x|^2 d\mu(x) = -\int |x|^2 d\mu(x) \leq 0.
   \]
   Since the image of the weakly continuous mapping $\pi \mapsto \int \langle x, y-x \rangle d\pi(x,y)$ over the connected set $\Pi(\mu,\nu)$ is itself connected, we conclude that \eqref{weakconvexorder} is equivalent to
   \be\label{weakconvexorder3}
   D(\mu, \nu) := \sup_{\pi \in \Pi(\mu,\nu)} \int \langle x, y-x \rangle d\pi(x,y) \geq 0.
   \ee
\eqref{weakconvexorder3} motivates the term {\em Kantorovich dominance}, as it shows that the "Kantorovich cost" \(\sup_{\pi \in \Pi(\mu,\nu)} \int \langle x, y \rangle d\pi\) dominates (i.e., is greater than or equal to) the second moment of \(\mu\).
 
 Since $ \displaystyle \Wtq(\mu, \nu) = \int |x|^2 d\mu + \int |y|^2 d\nu - 2\sup_{\pi \in \Pi(\mu,\nu)} \int \langle x, y \rangle d\pi$, \eqref{weakconvexorder3} is equivalent to
   \be\label{eqn: weak convex order with Wass}
   \Wtq(\mu, \nu) \leq \int |y|^2 d\nu - \int |x|^2 d\mu.
   \ee
\end{remark}

It follows immediately from \eqref{eqn: weak convex order with Wass} that the set of probability measures dominated by $\nu$
\[
\cM_{\nu, \Omega}^K := \{ \mu \in \cP_2(\Omega) \mid \mu \K \nu \}
\]
is closed with respect to $\Wt$-metric for any compact $\Omega \subset \R^d$. However, in general, the set
\[
    \cM_\nu^K := \{\mu \in \mathcal{P}_2(\R^d) \mid \mu \K \nu\}
\]
is not closed with respect to the $\Wt$ metric (see later discussion). 

\begin{remark}
    Let $B_\mu(r) = \left\{ \rho  \mid \Wt(\mu, \rho) \leq r \right\}$.
    \cite{WieselZhang2023} showed that $\mu \C \nu$     if and only if 
        \be\label{Wiesel}
        \Wtq(\nu, \rho) - \Wtq(\mu, \rho) \leq \int |y|^2 d\nu - \int |x|^2 d\mu 
    \ee
    holds for all $\rho \in B_\mu(\infty)=\mathcal{P}_2(\mathbb{R}^d)$. In contrast, $\mu \K \nu$ if \eqref{Wiesel} holds for $\rho \in B_\mu(0)$, i.e., for $\rho=\mu$. Exploring intermediate radii $r \in (0, \infty)$ remains an open direction for further study.
\end{remark}

We observe that the tail probability of $\mu$ is controlled by the second moment of $\nu$ if $\mu \K \nu$.

\begin{lemma}\label{lemma:moment_bound}
If $\mu \K \nu$, then we have
    \be
 \mu(\{|x|^2 \geq \alpha\}) \leq \frac{1}{\alpha}\int_Y |y|^2 d\nu \quad \text{for every } \ \alpha >0. \label{moment}
    \ee
\end{lemma}
 \eqref{moment} easily follows from \eqref{eqn: weak convex order with Wass} and Markov's inequality:
\[ \int |y|^2 d\nu \geq \Wtq(\mu, \nu) + \int |x|^2 d\mu 
        \geq \int |x|^2 d\mu 
        \geq \alpha \mu(\{|x|^2 \geq \alpha\}). 
        \]
If we further assume that $\Omega$ is compact and $\nu(\Omega)=1$ (e.g., $\Omega = \text{con}(\spt(\nu))$ with compact $\spt(\nu)$, where $\text{con}(\Omega)$ denotes the convex hull of $\Omega$), then we obtain a stronger compactness.

\begin{corollary}\label{cor:Kan_order_compact}
    $\cM_\nu^K$ is a precompact subset of $\cP(\R^d)$ with respect to the weak topology. If $\Omega$ is compact, then $\cM_{\nu, \Omega}^K$ is compact in both the weak and $\Wt$ topologies.
\end{corollary}

\begin{proof}
    By Lemma \ref{lemma:moment_bound} and Prokhorov's theorem, $\cM_\nu^K$ is precompact in the weak topology. Since ${\mathcal W}_2$ convergence is equivalent to weak convergence when $\Omega$ is compact (see Remark 2.8 of \cite{ambrosio2013}), and $\cM_{\nu, \Omega}^K$ is closed, the result follows.
\end{proof}

However, $\cM_\nu^K$ is not necessarily compact in the ${\mathcal W}_2$ topology, even if $\spt(\nu)$ is compact. This motivates the consideration of the domain $\cM_{\nu, \Omega}^K$with a compact $\Omega$ instead of $\cM_\nu^K$.

\begin{example}
   Let $\si_c \in \cP(\R^d)$ denote the uniform probability measure over a sphere with radius $c \ge 0$ centered at the origin. For $0 \leq r \leq 1 \leq R$, let $\mu = (1-\la)\si_r + \la \si_R$ and $\nu = \si_1$. The optimal coupling $\pi$ for $\Wtq(\mu, \nu)$ sends $(1-\la)\si_r$ radially to $(1-\la)\si_1$ and $\la \si_R$ to $\la \si_1$. 

    If $r \in (0,1)$ is fixed, and let
        $R = \frac{1 + \sqrt{1 + 4(r-r^2)(\frac{1}{\la} - 1)}}{2}$,
    it is straightforward to verify that inequality \eqref{eqn: weak convex order with Wass} still holds, and hence $\mu \K \nu$. This shows that the support of $\mu$ can be arbitrarily large (by letting $\lambda \to 0$), and that the Kantorovich order relation is not equivalent to the convex order relation, even in one dimension.

    This example also demonstrates that $\cM_\nu^K$ is not $\Wt$ compact in $\cP_2(\R^d)$, even if $\nu$ is compactly supported. If we let $r = 1/2$, $R$ as above, and $\mu_\la = (1-\la)\si_r + \la\si_R$, then as $\la \to 0$, $\mu_\la$ converges weakly to the uniform probability measure $\mu$ over a sphere with radius $1/2$ and second moment $1/4$. However, as $\la \to 0$, we have $\int |x|^2 d\mu_\la \to 1/2 \neq 1/4$.
\end{example}

Although the KDR reflects certain aspects of the convex order, it is not itself a partial order, as it generally fails to satisfy transitivity. We illustrate this with examples in one and two dimensions.

\begin{example}
  Let $\mu = \frac{1}{3} \delta_{-4} + \frac{2}{3} \delta_{2}$, $\nu = \frac{1}{3} \delta_{-4} + \frac{1}{3} \delta_{0} + \frac{1}{3} \delta_{4}$ and $\theta = \frac{2}{3} \delta_{-3} + \frac{1}{3} \delta_{6}$. Then $\Wtq(\mu, \nu) = \frac{8}{3}$, ${\mathcal W}_2^2(\nu, \theta) = \frac{14}{3}$ and ${\mathcal W}_2^2(\mu, \theta) = 14$. On the other hand, $\int |x|^2 d\mu = 8$, $\int |y|^2 d\nu = \frac{32}{3}$ and $\int |z|^2 d\theta = 18$. By \eqref{eqn: weak convex order with Wass}, $\mu \K \nu$ and $\nu \K \theta$, but we do not have $\mu \K \theta$.
\end{example}

\begin{example}
 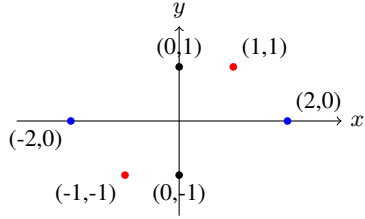
\begin{figure}
        \centering
        \begin{tikzpicture}[scale=0.72] 
            % Draw axes
            \draw[->] (-3,0) -- (3,0) node[right] {$x$};
            \draw[->] (0,-1.75) -- (0,1.75) node[above] {$y$};
            
            % Draw points with colors
            \fill (0,1) circle (2pt);
            \fill[red] (1,1) circle (2pt);
            \fill[blue] (2,0) circle (2pt);
        
            \fill (0,-1) circle (2pt);
            \fill[red] (-1,-1) circle (2pt);
            \fill[blue] (-2,0) circle (2pt);
        
            % Display coordinates in black
            \node[above] at (0,1) {(0,1)};
            \node[above right] at (1,1) {(1,1)};
            \node[above right] at (2,0) {(2,0)};
        
            \node[below] at (0,-1) {(0,-1)};
            \node[below left] at (-1,-1) {(-1,-1)};
            \node[below left] at (-2,0) {(-2,0)};
        
        \end{tikzpicture}
        \caption{An example showing that the KDR is not transitive and hence not a partial order.}
        \label{fig:2d_countex}
    \end{figure}
    Let $\mu = \frac{1}{2} \delta_{(0,-1)} + \frac{1}{2} \delta_{(0,1)}$, $\nu = \frac{1}{2} \delta_{(-1,-1)} + \frac{1}{2} \delta_{(1,1)}$, and $\theta = \frac{1}{2} \delta_{(-2,0)} + \frac{1}{2} \delta_{(2,0)}$ (see Figure \ref{fig:2d_countex}). Clearly, $\Wtq(\mu, \nu) = 1$ and ${\mathcal W}_2^2(\nu, \theta) = 2$. Meanwhile, $\int |x|^2 d\mu = 1$, $\int |y|^2 d\nu = 2$, and $\int |z|^2 d\theta = 4$. Thus, $\mu \K \nu$ and $\nu \K \theta$ by \eqref{eqn: weak convex order with Wass}. However, since $\Wtq(\mu, \theta) = 5 > 3 =  \int |z|^2 d\theta - \int |x|^2 d\mu$, we do not have $\mu \K \theta$.
\end{example}

\section{Maximizing variance under the Kantorovich dominance relation}\label{sec4}

Throughout the sequel, we will assume that $\Omega$ is a compact, convex set containing $0$, unless stated otherwise. An example is $\Omega = \text{con}(\spt(\nu))$. Let $\nu \in \cP_{2,0}(\Omega)$.  Motivated by our earlier discussion, given a domain $\cD \subset \cP_2(\R^d)$, we now study the following variance maximization problem:
\be\label{weakproblem}
    \max_{\mu \in \cD \cap \cM_{\nu, \Omega}^K} \Var (\mu).
\ee

\begin{theorem}
  %  \leavevmode\newline
    \begin{enumerate}[label=\roman*), ref=\roman*]
        \item  If $\cD \subset \cP_2(\R^d)$ is closed in the ${\mathcal W}_2$-metric and $ \cD \cap \cM_{\nu, \Omega}^K$ non-empty, then problem \eqref{weakproblem} admits a solution.

        \item Let $\mu^*$ be any solution to \eqref{weakproblem}. Then for any  $\rho \in \cD \cap \cM_\nu^K \cap \cP_{2,0}(\R^d)$, we have:
        \[
            \Wt(\mu^*, \rho) \leq \sqrt{\int |y|^2 d\nu - \int |z|^2 d\rho} + \Wt(\nu, \rho).
        \]
        Consequently, problem \eqref{weakproblem} exhibits a \textit{denoising property}, meaning that as the noise in the measure $\nu$ decreases, the optimal solution $\mu^*$ tends to recover the original measure.
    \end{enumerate}
\end{theorem}

\begin{proof}
    i) Since $\cM_{\nu, \Omega}^K$ is compact in the ${\mathcal W}_2$-metric by Corollary \ref{cor:Kan_order_compact}, and $\cD$ is ${\mathcal W}_2$-closed, it follows that $\cD \cap \cM_{\nu, \Omega}^K$ is also ${\mathcal W}_2$-compact. The continuity of the variance functional with respect to the ${\mathcal W}_2$-metric guarantees that problem \eqref{weakproblem} admits a solution.

    ii) The proof is similar to the proof of Theorem \ref{existence}.
    \begin{align*}
        \Wt (\mu^*, \rho) &\leq \Wt(\mu^*, \nu) + \Wt(\nu, \rho) \\
        &\leq \sqrt{\int |y|^2 d\nu - \int |x|^2 d\mu^*} + \Wt(\nu, \rho)
        \leq \sqrt{\int |y|^2 d\nu - \int |z|^2 d\rho} + \Wt(\nu, \rho),
    \end{align*}
    where the second inequality is by the KDR and the last inequality follows from the optimality of $\mu^*$, which gives $\E_\rho [|X|^2] = \Var(\rho) \le \Var(\mu^*) = \E_{\mu^*} [|X|^2] - | \E_{\mu^*} [X] |^2 \le  \E_{\mu^*} [|X|^2]$.
\end{proof}

\subsection{Stability}
Under a certain condition on $\cD$, we now show problem \eqref{weakproblem} is stable, in contrast to \eqref{problem}, as shown in Example \ref{counterexample:stability}.  The hypothesis we will impose on $\cD$ in the following definitions is quite weak; in particular, it is satisfied by all domains we have considered.  

\begin{definition}\label{def: approachability}
   Let $\nu \in \cP_{2,0}(\R^d)$.  We say that $\cD$ is {\em approachable from the interior with respect to $\nu$} if, for any $\mu \in \cD$ with $\mu \neq \delta_0$ and $\mu \K \nu$, there exists a sequence $\{ \mu_n \} \subset \cD$ such that $\Wt(\mu_n, \mu) \to 0$ and $D(\mu_n, \nu) > 0$ in \eqref{weakconvexorder3} for all $n \in \N$. 
\end{definition}
The class of domains that are closed under contractions serves as a key example of domains that can be approached from the interior.

\begin{definition}
Let $\la_\# \mu$ denote the \textit{dilation} of $\mu$ by $\la \in \R$, i.e., $\la X \sim \la_\# \mu$ when $X \sim \mu$. 
We say that $\cD$ is \textit{closed under contractions} if $\la_\# \mu \in \cD $ for any $\mu \in \cD$ and $\la \in (0,1)$.
\end{definition}

\begin{lemma}
If $\cD$ is closed under contractions, then for any $\nu \in \mathcal{P}_{2,0}(\Omega)$ and  $\delta_0 \neq \mu \in \cD$ with $\mu \K \nu$, there is a sequence $\{\mu_n\} \subset \cD$ such that  $W_2(\mu_n,\mu) \rightarrow 0$ and $D(\mu_n, \nu) > 0$.
\end{lemma}
\begin{proof}
If $ D(\mu,\nu) > 0$ there is nothing to prove.  If not, there exists some $\pi \in \Pi(\mu,\nu)$ with $\int \langle x, y-x \rangle d\pi(x,y) =0$. Let $\lambda_n$  be an increasing sequence of positive numbers converging to $1$ and choose $\mu_n =\lambda_{n\#}\mu$ and $\pi_n = \big((x,y) \mapsto (\lambda_nx,y)\big)_\# \pi$, so that $\int \langle x, y-x \rangle d\pi_n(x,y) =\lambda_n\int \langle x, y \rangle d\pi(x,y)- \lambda_n^2\int|x|^2d\mu(x)  > \lambda_n\int \langle x, y \rangle d\pi(x,y)- \lambda_n\int|x|^2d\mu(x)=0$.
\end{proof}

The following result illustrates a stability property inherent to problem \eqref{weakproblem}, where $\Omega \subset \R^d$ is not required to be compact.

\begin{theorem}\label{thm:stability}
 Let $\nu \in \mathcal{P}_{2,0}(\Omega)$ and assume that $\cD$ is approachable from the interior w.r.t. $\nu$. Let $\{\nu_n\} \subset \cP_2(\Omega)$ such that ${\mathcal W}_2(\nu_n, \nu) \to 0$. Let $\mu_n$ be a solution to \eqref{weakproblem} with $\nu_n$. If ${\mathcal W}_2(\mu_n, \mu) \to 0$ for some $\mu$, then $\mu$ is a solution to \eqref{weakproblem} with $\nu$.
\end{theorem}

\begin{proof}
    Consider first an arbitrary  $\xi \in \cD \cap \cM_{\nu, \Omega}^K$ satisfying ${\mathcal W}_2^2(\xi, \nu) < \int |y|^2 d \nu(y) - \int |x|^2 d \xi(x)$. By the strict inequality, for all large $n$ we have $\xi \in \cD \cap \cM_{\nu_n,\Omega}^K$.  By optimality of $\mu_n$ in \eqref{weakproblem} with $\nu_n$, we have $\Var(\mu_n) \geq \Var(\xi)$; taking limits then yields $\Var(\mu) \geq \Var(\xi)$.

    Now consider any $\xi \in \cD \cap \cM_{\nu,\Omega}^K$.  Letting $\{\xi_m\} \subset \cD \cap \cM_{\nu,\Omega}^K$ approximate $\xi$ such that  ${\mathcal W}_2^2(\xi_m, \nu) < \int |y|^2 d \nu - \int |x|^2 d \xi_m$, the above argument yields  $\Var(\mu) \geq \Var(\xi_m)$; taking limits implies  $\Var(\mu) \geq \Var(\xi)$.  As $\xi \in \cD \cap \cM_{\nu,\Omega}^K$ is arbitrary, optimality of $\mu$ follows. 
\end{proof}

\begin{remark}
    A natural question is whether a quantitative version of the approachable from the interior condition yields quantitative stability of \eqref{weakproblem}.  Namely, suppose we  require that for every $\mu \neq \delta_0$ in $\cD \cap\cM_{\nu,\Omega}^K$ the existence of an approximating sequence $\{\mu_n\} \subset \cD$ as in Definition \ref{def: approachability}  with the additional condition that $f({\mathcal W}_2(\mu_n, \mu)) < f(D(\mu_n,\nu)) $, for some suitable strictly increasing continuous function $f :[0,\infty) \rightarrow [0,\infty)$, with $f(0)=0$.  Does this then imply that $g({\mathcal W}_2(\hat \mu, \tilde \mu)) < g( {\mathcal W}_2(\hat \nu, \tilde \nu))$, where $\hat \mu$ and $\tilde \mu$ are solutions in \eqref{weakproblem} for $\hat \nu$ and $\tilde \nu$, respectively, for an appropriate, strictly increasing continuous functions $g :[0,\infty) \rightarrow [0,\infty)$, with $g(0)=0$?

    Like related issues regarding the quantitative stability of optimal transport (see for instance \cite{DelalandeMerigot23}), this question seems very challenging.  Our proof of Theorem \ref{thm:stability} does not seem to provide a particular insight into how one might approach the quantitative stability question, and its resolution will likely require significant new ideas and techniques.
\end{remark}

\subsection{Conic domains and equivalence with the relaxed problem}\label{sec5}

We establish the equivalence between the problem \eqref{weakproblem} and the relaxed problem \eqref{eqn: relaxed problem} under certain structural assumptions on $\cD$. Notably, this result applies even when $\spt(\nu) = \Omega$ is not compact, ensuring the existence and stability of solutions in such cases as well.

\begin{definition}
   $\cD  \subset \cP(\R^d)$ is called a \textit{cone} if $\lambda_\# \mu \in \cD$ for any $\mu \in \cD$ and $\lambda \geq 0$.   $\cD$ is called \textit{translation invariant} if ${T_k}_\# \mu \in \cD$ for any $\mu \in \cD$ and $k \in \R^d$, where $T_k(x):= x + k$.
\end{definition}
It is worth noting that problem \eqref{weakproblem} is stable on cone domains, as these domains are closed under contraction.

Many of the domains we are interested in are cones and translation invariant, including those in Examples \ref{ex: monotone} and \ref{ex: discrete}, as well as measures supported on lines and multivariate Gaussians as will be explored in Section \ref{applications} below.  Note that the domain in Example \ref{ex: lipschitz} is not a cone.

\begin{lemma}\label{lemma:Kantorovich_same_mean}
    Let $\mu, \nu \in \cP_2(\mathbb{R}^d)$, and let $\mu_k := {T_k}_\# \mu$ denote the translation of $\mu$ by $k$. Among all translations $\mu_k$, the one that minimizes the $\Wt$ distance to $\nu$ is the one for which
    \be\label{samecenter}
    \int_{\mathbb{R}^d} x \, d\mu_k(x) = \int_{\mathbb{R}^d} y \, d\nu(y).
    \ee
\end{lemma}

\begin{proof}
    The $\Wt$ distance between $\mu_k$ and $\nu$ is given by
    \[
    \Wtq(\mu_k, \nu) = \int_{\R^d \times \R^d} |(x + k) - y|^2 d\pi(x,y),
    \]
    where $\pi$ is an optimal transport plan between $\mu$ and $\nu$. Differentiating with respect to $k$ and setting the derivative equal to zero gives $
    k = \int y d\nu(y) - \int x d\mu(x)$. 
    Thus, the translation $\mu_{k^*}$ that minimizes the $\Wt$ distance satisfies \eqref{samecenter}, as claimed.
\end{proof}

\begin{lemma}\label{lemma:optimality_cone}
   Let \(\nu \in \cP_{2,0}(\R^d)\) and \(\cD\) be a cone. For any $\mu^* \in \argmax_{\mu \in \cD \cap \cM_\nu^K} \Var(\mu)$ and $ \pi^* \in \argmax_{\pi \in \Pi(\mu^*, \nu)} \int \langle x, y \rangle d\pi (x,y)$, the following equality holds:
    \be\label{eqn:kantorovich_opt}
    \int \langle x, y \rangle d\pi^*(x,y) = \int |x|^2 d\mu^*(x).
    \ee
    If, in addition, $\cD$ is translation invariant, then $\mu^*$ must be centered; $\int x\, d\mu^*(x)=0$.
\end{lemma}

\begin{proof}
    By the KDR \eqref{weakconvexorder3}, for any $\mu \in \cD \cap \cM_\nu^K$ and for any corresponding optimal coupling $\pi^* \in \argmax_{\pi \in \Pi(\mu, \nu)} \int \langle x, y \rangle d\pi(x,y)$, we have the inequality:
    \be\label{eqn:kantorovich_ineq}
    \int \langle x, y \rangle d\pi^*(x,y) \geq \int |x|^2 d\mu(x).
    \ee

    Suppose this inequality is strict for some $\mu^* \in \argmax_{\mu \in \cD \cap \cM_\nu^K} \Var(\mu)$. As $\cD$ is a cone,  $\mu^*_\la = \la_{\#} \mu^* \in \cD$ for any $\la \ge 0$. Let  \(\pi^*_\lambda = \big( (x,y) \mapsto (\lambda x, y)\big)_\# \pi^*\) be the corresponding coupling of $\mu^*_\la$ and $\nu$. 
    If the inequality \eqref{eqn:kantorovich_ineq} is strict, there exists a \(\lambda > 1\) such that:
    \[
    \int \langle x, y \rangle d\pi^*_\lambda = \lambda \int \langle x, y \rangle d\pi^* \geq \lambda^2 \int |x|^2 d\mu^* = \int |x|^2 d\mu^*_\lambda,
    \]
    meaning that $\mu^*_\lambda \K \nu$. However, the variance of \(\mu^*_\lambda\) exceeds that of \(\mu^*\) (unless $\mu^* = \delta_0$, in which case the result is trivial), contradicting the optimality of $\mu^* \in \argmax_{\mu \in \cD \cap \cM_\nu^K} \Var(\mu)$. We conclude that the inequality \eqref{eqn:kantorovich_ineq} must hold with equality for any optimal \(\mu^*\) and \(\pi^*\).

     To see $\int x d \mu^* = 0$ if $\cD$ is translation invariant, let $k = -\int x d \mu^*$ and $\mu^*_{k} = {T_{k}}_\# \mu^*$ be defined as in Lemma \ref{lemma:Kantorovich_same_mean}. Setting a coupling $\pi^*_{k} = \big((x,y) \mapsto (x + k, y)\big)_\# \pi^*$ of $\mu^*_{k}$ and $\nu$ and recalling that $\nu$ is centered, the following calculation
      \[
    \int |x|^2 d\mu_{k}^* = \int |x + k|^2 d\mu^* \leq \int |x|^2 d\mu^* = \int \langle x, y \rangle d\pi^* = \int \langle x + k, y \rangle d\pi^* = \int \langle x, y \rangle d\pi^*_{k}
    \]
     shows $\mu_{k}^* \K \nu$. With this,   $\Var(\mu_{k}^*) = \Var(\mu^*)$ implies $\mu_k^*$ is also optimal. Now $k \neq 0$ yields strict inequality above, which contradicts \eqref{eqn:kantorovich_opt}. We conclude $k=0$.
\end{proof}

We can now extend the solution existence result for the problem \eqref{weakproblem} when $\Omega = \R^d$.

\begin{theorem}\label{thm: kant equivalence to relaxed prob}
    Let \(\nu \in \cP_{2,0}(\R^d)\), and let \(\cD\) be a cone. Suppose that \(\cD\) is either translation invariant or \(\cD \subset \cP_{2,0}(\R^d)\). Then, the problem \eqref{weakproblem} with \(\Omega = \R^d\) is equivalent to the relaxed problem \eqref{eqn: relaxed problem}, meaning that both formulations have the same set of solutions.
\end{theorem}

\begin{proof}
For any $\displaystyle \mu^* \in \argmin_{\mu \in \cD}{\mathcal W}_2^2(\mu,\nu)$ and $ \displaystyle \pi^* \in \argmax_{\pi \in \Pi(\mu^*, \nu)} \int \langle x, y \rangle d\pi(x,y)$, we claim
    \be\label{eqn: projection}
        \int \langle x, y - x \rangle \, d \pi^*(x,y) = 0,
    \ee
    that is, $\mu^* \K \nu$. To see this, for any $\mu \in \cD \setminus \{ \delta_0 \}$, consider a rescaling factor $\lambda^*$ that solves
    \[
    \min_{\lambda \in \R} {\mathcal W}_2^2(\lambda_{\#} \mu, \nu) = \min_{\lambda \in \R} \left( \lambda^2 \int |x|^2 d \mu - 2 \lambda \max_{\pi \in \Pi(\mu, \nu)} \int \langle x, y \rangle d \pi(x,y) + \int |y|^2 d \nu \right).
    \]
    The optimal scaling is attained at:
    \[
    \lambda^* = \frac{\max_{\pi \in \Pi(\mu, \nu)} \int \langle x, y \rangle d \pi(x,y)}{\int |x|^2 d \mu}.
    \]
    Since $\lambda^* = 1$ for the optimal $\mu^*$ (no rescaling improves the objective function), \eqref{eqn: projection} follows.

   Next, by the assumption on $\cD$ and Lemma \ref{lemma:Kantorovich_same_mean}, $\int x d \mu^* = \int y d \nu = 0$. This and \eqref{eqn: projection} imply
    \begin{eqnarray}\label{eqn: equality of values}
        \min_{\mu \in \cD} \Wtq(\mu, \nu) &=& \min_{\mu \in \cD \cap \cM_\nu^K \cap \cP_{2,0}(\R^d)} \Wtq(\mu, \nu)\nonumber\\
        &=& \min_{\mu \in \cD \cap \cM_\nu^K \cap \cP_{2,0}(\R^d)} \left( \int |y|^2 d \nu - \int |x|^2 d \mu \right) \nonumber\\
        &=& \Var(\nu) - \max_{\mu \in \cD \cap \cM_\nu^K} \Var(\mu),
    \end{eqnarray}
    where we have used Lemma  \ref{lemma:optimality_cone} to remove $\bigcap \cP_{2,0}(\R^d)$ in \eqref{eqn: equality of values}. This shows any solution to the relaxed problem \eqref{eqn: relaxed problem} is also a solution to  \eqref{weakproblem} with $\Omega= \R^d$.  Conversely, if $\mu^*$ solves \eqref{weakproblem}, the equality of values \eqref{eqn: equality of values} and Lemma \ref{lemma:optimality_cone} imply that $\mu^*$ solves \eqref{eqn: relaxed problem}, completing the proof.
\end{proof}

\begin{remark}
Observe that the proofs of Lemma \ref{lemma:optimality_cone} and Theorem \ref{thm: kant equivalence to relaxed prob} do not rely on full translation invariance of $\mathcal D$. Instead, the following weaker condition suffices (recalling that these results assume $\nu$ is centered): for any $\mu \in \mathcal D$ with $k = - \int x d\mu(x)$, the centered measure $\mu_k = {T_k}_\# \mu \in \cP_{2,0}(\R^d)$ also belongs to $\mathcal D$.
\end{remark}

\subsection{Weak optimizer closedness and equivalence to the problem with self-consistency}

We have shown that for appropriate domains, the problem \eqref{weakproblem} is equivalent to the fully relaxed problem \eqref{eqn: relaxed problem}. We now consider when \eqref{weakproblem} is equivalent to the problem with the full self-consistency condition \eqref{problem}. 
Recall the $\pi$-conditional barycentric map from Definition \ref{def: centering map}.
\begin{definition}
Let $\nu \in \cP_{2,0}(\Omega)$ and $\cD \subset \cP_2(\Omega)$.
 We say that $\cD$ is  closed under weak optimizers if, 
 for any $\mu \in \cD \cap \cM^K_\nu$ that solves \eqref{weakproblem}, there exists $\pi \in \Pi(\mu,\nu)$ such that
\be\label{recenter}
 \int \langle x, y \rangle d\pi(x,y) \geq \int |x|^2 d\mu(x)  \quad \text{and} \quad {c_\pi}_\#(\mu) \in \cD.
\ee
\end{definition}
For instance, the domains $\cD^m$ and $\cD^m_u$ in Example \ref{ex: discrete} are closed under weak optimizers.

\begin{theorem}\label{equivalencetheorem}
  If $\cD$ is closed under weak optimizers, then every optimizer $\mu$ for the problem \eqref{weakproblem} satisfies $\mu \C \nu$, and consequently, $\mu$ solves the original problem \eqref{maxvar}.
\end{theorem}

\begin{proof}
Assume that $\mu$ solves \eqref{weakproblem}, so there exists $\pi \in \Pi(\mu,\nu)$ satisfying \eqref{recenter}. Since $\gamma = \big( (x,y) \mapsto (c_\pi(x), y)\big)_\# \pi \in \cM({c_\pi}_\#(\mu), \nu)$ is a martingale measure, we have
\[
\int |x|^2 d {c_\pi}_\#(\mu) - \int \langle x, y \rangle d\gamma = 0 \geq \int |x|^2 d\mu - \int \langle x, y \rangle d\pi.
\]

Now,  as the barycenter minimizes the expected squared distance, we have the inequality
\[
\int \frac{1}{2} |x-y|^2 d\gamma(x,y) \leq \int \frac{1}{2} |x-y|^2 d\pi(x,y),
\]
with strict inequality if ${c_\pi}_\#(\mu) \ne \mu$. After canceling $\int |y|^2 d\gamma = \int |y|^2 d\pi = \int |y|^2 d\nu$, we get
\[
\int \left( -\frac{|x|^2}{2} + \langle x, y \rangle \right) d\gamma \geq \int \left( -\frac{|x|^2}{2} + \langle x, y \rangle \right) d\pi,
\]
with strict inequality if ${c_\pi}_\#(\mu) \ne \mu$. Adding this to the previous inequality gives
\[
\Var({c_\pi}_\#(\mu)) = \int |x|^2 d{c_\pi}_\#(\mu) \geq \int |x|^2 d\mu \geq \Var(\mu).
\]
The optimality of $\mu$ implies equality. Hence ${c_\pi}_\#(\mu) = \mu$, which implies $\mu \C \nu$.
\end{proof}

\subsection{Relationship with principal component analysis}\label{applications}  
In this section, we explore the relationship between our formulation \eqref{weakproblem} and principal component analysis (PCA). Let $V_m$ denote the set of all $m$-dimensional subspaces of $\R^d$. Consider the following cone domain:
\[
\cD_m := \{ \mu \in \cP_{2,0}(\R^d) \, | \, \mu(E) =1 \text{ for some } E \in V_m \}.
\]
Let $p_E : \R^d \to E$ denote the orthogonal projection map onto a subspace $E$ of $\R^d$. 

\begin{lemma}\label{PCAlemma}
Any $\mu$ solving the problem $\displaystyle \max_{\mu \in \cD_m, \, \mu \K \nu}  \Var (\mu)$, where $\nu \in \cP_{2,0}(\R^d)$, is the orthogonal projection of $\nu$ onto some $E \in V_m$.
\end{lemma}

\begin{proof} 
Theorem \ref{thm: kant equivalence to relaxed prob} shows the problem $ \max_{\mu \in \cD_m, \, \mu \K \nu}  \Var (\mu)$ is equivalent to the relaxed problem $ \min_{\mu \in \cD_m} {\cal W}_2^2(\mu,\nu)$. Noting $\cD_m = \cup_{L \in V_m} \cD_m^E$, where $\cD_m^E := \{ \mu \in \cD_m \, | \, \mu(E) =1 \}$, we can decompose the relaxed problem as $\min_{E \in V_m} \min_{\mu \in \cD_m^E} \Wtq(\mu,\nu)$. The lemma follows by the fact that the projection ${p_E}_\# \nu$ is the unique minimizer of $ \min_{\mu \in \cD_m^E} \Wtq(\mu,\nu)$.
\end{proof}

Lemma \ref{PCAlemma} reveals that PCA can be viewed as a particular case of problem \eqref{weakproblem} with the domain $\cD_1$. Specifically, the first principal component is defined as the direction that maximizes the variance of the projected data. Lemma \ref{PCAlemma} demonstrates that the projected data satisfies the variance maximization problem under the Kantorovich dominance constraint.
\\

We now consider the following version of the PCA in \cite[Section 3.3]{ChenChiFanMa2021Spec}. Consider the model
\be
Y = L^*W + R,
\ee
where $W \sim  \cN(0, I_m)$ is an $m$-dimensional Gaussian vector of latent factors, $L^* \in \R^{d \times m}$ is an unknown factor loading matrix of rank $m$, and $R \sim \cN(0, \sigma^2 I_d)$ represents random noise that is independent of $W$ and cannot be explained by the latent factor. Without loss of generality, assume that $L^* = U^*(\Lambda^*)^{1/2}$, where the columns of $U^* \in \R^{d \times m}$ form an orthonormal set, and $\Lambda^* = {\rm diag}[\lambda_1^*, \dots, \lambda_m^*]$ is a diagonal matrix with $\lambda_1^* \geq \dots \geq \lambda_m^* > 0$.

Let $\nu = \cL(Y)$. We now focus on solving the problem \eqref{weakproblem} over the following cone domain
\begin{align}\label{PCAD}
\cD = \{\mu_L &=  \cL(LW) \mid L = U\Lambda^{1/2}, \text{ where the columns of } U \in \R^{d \times m} \text{ are orthonormal,} \\
&\text{and } \Lambda = {\rm diag}[\lambda_1, \dots, \lambda_m] \text{ is a diagonal matrix with } \lambda_1 \geq \dots \geq \lambda_m \geq 0 \}. \nn
\end{align}
Note that the problem \eqref{weakproblem} remains equivalent to the relaxed problem \eqref{eqn: relaxed problem}, by Theorem \ref{thm: kant equivalence to relaxed prob}.

\begin{theorem}\label{PCADthm}
Let $\nu_n, \nu \in \cP_{2,0}(\R^d)$ and $\Wt(\nu_n,\nu) \to 0$. Then for all $n$, $\cD \cap \cM_{\nu_n}^K \neq \emptyset$ with $\cD$ in \eqref{PCAD}, and for any $\displaystyle \mu_{L_n} \in \argmax_{\mu \in \cD \cap \cM_{\nu_n}^K} \Var(\mu)$ with $L_n = U_n \Lambda_n^{1/2}$, 
\be\label{recovery}
L_n L_n^\tr - \sigma_n^2 U_n U_n^\tr \to L^* {L^*}^\tr \quad \text{and} \quad \sigma_n^2 \to \sigma^2 \quad \text{as} \quad n \to \infty,
\ee
where $\sigma_n^2 := \frac{1}{d-m} \int |y - U_n U_n^\tr y |^2 d \nu_n(y)$ is an estimator of the noise variance $\sigma^2$. 
\end{theorem}

\begin{remark}\label{rem: denoising with convex order is trivial} \eqref{recovery} indicates that the optimization problem $ \max_{\mu \in \cD \cap \cM_{\nu_n}^K} \Var(\mu)$ can recover $\cL(L^*W)$ as the empirical distribution $\nu_n$ converges to the population distribution $\nu$. 
This result cannot hold if the full convex order constraint is imposed. In many practical cases, such as when \(\nu_n\) is discrete (e.g., \(\nu_n\) is an empirical measure sampled from \(\nu\)), the convex order condition \(\mu \C \nu_n\) fails for any Gaussian \(\mu \in \cD\) unless $\mu = \delta_0$. As a result, the domain \(\cD \cap \{\mu \C \nu_n\}\) reduces to \(\{\delta_0\}\).
\end{remark}

\begin{lemma}\label{PCADcompact}
For $\cD$ in \eqref{PCAD}, the set $\cD \cap \{ \mu \mid \Var(\mu) \le \delta \}$ is $\mathcal{W}_2$-compact for any $\delta >0$.
\end{lemma}

\begin{proof}
For any $\mu_L \in \cD$ with $L = U\Lambda^{1/2}$, $\Var(\mu_L) = \sum_{i=1}^m \lambda_i$. Thus, $\max_i \lambda_i \leq \delta$.
\end{proof}

\begin{lemma}\label{NormDistConvLemma}
Let $\mu_n = \mathcal{L}(U_n \Lambda_n^{1/2} W)$ and $\mu = \mathcal{L}(U \Lambda^{1/2} W)$ as in \eqref{PCAD}. If $\mathcal{W}_2(\mu_n, \mu) \to 0$, then $U_n \Lambda_n U_n^\tr \to U \Lambda U^\tr$ and $\Lambda_n \to \Lambda$. If ${\rm rank}(\La)=m$, then $U_n U_n^\tr \to U U^\tr$ as well.
\end{lemma}

\begin{proof}
Assume that $\lambda_1 > \dots > \lambda_k > 0 = \lambda_{k+1} = \dots = \lambda_m$ in the diagonal of $\Lambda$. Since $\mathcal{W}_2(\mu_n, \mu) \to 0$ implies $\Var(\mu_n)$ is bounded, the sequences $\{U_n\}$ and $\{\Lambda_n\}$ are precompact. For any subsequences $\{U_k\}$ of $\{U_n\}$ and $\{\Lambda_k\}$ of $\{\Lambda_n\}$ converging to $\tilde{U}$ and $\tilde{\Lambda}$, respectively, define $\tilde{\mu} = \mathcal{L}(\tilde{U} \tilde{\Lambda}^{1/2} W)$. Then, $\mathcal{W}_2(\mu_n, \mu) \to 0$ implies $\tilde{\mu} = \mu$, which ensures 
$\tilde{U}_i \tilde{U}_i^\tr = U_i U_i^\tr$ for $i=1,\dots,k$ where $U_i$ is the $i$th column of $U$ (notice this yields $\tilde U \tilde{U}^\tr = U U^\tr$ if $\la_m >0$), 
and $\tilde{\Lambda} = \Lambda$, i.e., $\tilde{\lambda}_i = \lambda_i$ for $i = 1, \dots, m$. Consequently, $\tilde{U} \tilde{\Lambda} \tilde{U}^\tr = U \Lambda U^\tr$. The arbitrariness of the subsequence establishes the lemma. The more general case, $\lambda_1 \geq \dots \geq \lambda_k$, can be similarly proven by taking into account the multiplicity of the singular values $\la_1,\dots, \la_k$.
\end{proof}

\begin{proof}[Proof of Theorem \ref{PCADthm}]
$\cD \cap \cM_{\nu_n}^K$ is compact by Lemma \ref{PCADcompact}, and 
$\delta_0 \in \cD \cap \cM_{\nu_n}^K$. Notice the unique solution to $\max_{\mu \in \cD \cap \cM_{\nu}^K} \Var(\mu)$ is $
\mu^* = \cN(0, U^*(\Lambda^* + \sigma^2 I_m) {U^*}^T)$, 
since $\nu = \cN(0, U^* \Lambda^* {U^*}^T + \sigma^2 I_d)$, and the $\mathcal{W}_2$-projection of $\nu$ onto $\cD$ is clearly $\mu^*$. 
Then for any $ \mu_{L_n} \in \argmax_{\mu \in \cD \cap \cM_{\nu_n}^K} \Var(\mu)$, Lemma \ref{PCADcompact} and Theorem \ref{thm:stability} give $\Wt(\mu_{L_n}, \mu^*) \to 0$. Then since $\mu_{L_n} = \cL(U_n \Lambda_n^{1/2} W)$ and $\mu^* = \cL(U^*(\Lambda^* + \sigma^2 I_m)^{1/2} W)$, Lemma \ref{NormDistConvLemma} gives $U_n \Lambda_n U_n^\tr \to U^*(\Lambda^* + \sigma^2 I_m) {U^*}^\tr$, $\Lambda_n \to \Lambda^* + \sigma^2 I_m$ and $U_n U_n^\tr \to U^* {U^*}^\tr$. This with $\nu_n \to\nu$ implies $\sigma_n^2 \to \sigma^2 = \frac{1}{d-m}\int |y - U U^\tr y |^2 d \nu(y)$, yielding $L_n L_n^\tr - \sigma_n^2 U_n U_n^\tr \to U^* \Lambda^* {U^*}^\tr$.
\end{proof}

\section{Numerical examples}\label{sec_numerics}

We provide examples of numerically solving the KDR problem \(\displaystyle \max_{\mu \in \cD,\, \mu \K \nu} \Var (\mu)\) with $\nu = \frac1n \sum_{j=1}^n \delta_{y_j}$ and three closely related discrete curve domains. 

\subsection{Curves with bounded length}\label{ex_bounded_length}
\begin{figure}%[h]
    \centering
    \includegraphics[width=1\linewidth]{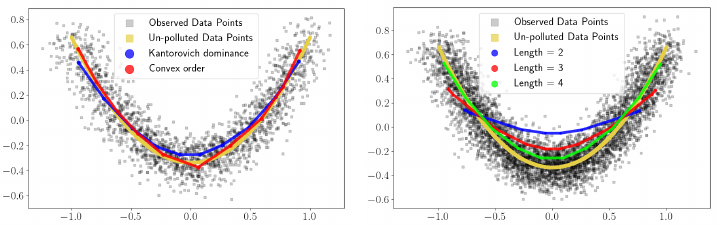}
    \caption{Left: Optimal measure $\mu$ under the Kantorovich dominance and the convex order for 2000 data points. Right: Optimal measure $\mu$ under the Kantorovich dominance for 5000 data points.}
    \label{fig:len_constr_compare}
\end{figure}
 Our first example considers the following domain
 \be
\cD_1 = \cD_1(m, B) = \bigg\{ \mu =  \sum_{i=1}^m u_i \delta_{x_i} \,\bigg|\, x_i \in \R^d,\, u_i \ge 0,\, \sum_{i=1}^m u_i = 1,\, L(\mu) \le B \bigg \}, \nn
\ee
where $L(\mu) = \sum_{i=1}^{m-1} \|x_{i+1} - x_i\|$ represents the length of the discrete curve $\mu$. Due to the length bound, $\cD_1$ is not a cone. 
To address the bound, we transform the constrained optimization problem into an unconstrained form by introducing a Lagrangian with multipliers that correspond to each constraint. Details are given in Section B of the Supplementary Material. The method employs gradient ascent to update the position  $x = (x_i)_{i=1}^m \in (\R^d)^m$ and weight $u = (u_i)_{i=1}^m$, while applying gradient descent to update the Lagrange multipliers iteratively. 

It is possible to numerically solve problem \eqref{problem}, particularly the martingale constraint, using gradient descent. However, this approach involves optimizing over a larger set of unknowns for the coupling $\pi$ rather than focusing on $\mu$ alone, which leads to increased memory requirements for storing variables, greater computational demand for calculating gradients, and potentially longer convergence times. In contrast, although calculating the coupling $\pi$ is still required to compute the $\Wt$ distance for the KDR problem, it can be done efficiently using the Sinkhorn algorithm \cite{Cuturi2013}, \cite{carlier2022linear}, which is known for its speed and computational efficiency. 

To illustrate our example, we consider $\nu$ as a discrete measure, defined over either $n=2000$ or $5000$ points, where $Y = (Z, Z^2) + \varepsilon$, with $Z$ being a one-dimensional variable uniformly distributed over $[-1, 1]$, and $\varepsilon$ representing Gaussian noise. The initial measure $\mu$ is supported on $m = 10$ points, which are initialized along the first principal component of $\nu$. 

Using the proposed method, we optimize the location $\bbx$ and the weight $\bbu$. 
In the case of $2000$ data points, we set the length bound $B = 4.0$ and the left side of Figure \ref{fig:len_constr_compare} illustrates the optimal $\mu$ under both the Kantorovich dominance relation and the convex order relation. For the case with 5000 data points, an out-of-memory error occurred when attempting to compute the optimal measure under the convex order. 
Therefore, we present only the optimal measure under the KDR, with $B = 2.0, 3.0$ and $4.0$, as shown in the right side of Figure \ref{fig:len_constr_compare}.
\footnote{The Python code for Example \ref {ex_bounded_length} is available at https://github.com/joshuaHiew/Kantorovich\_dominance.}
\\

Further improvement in performance is achieved in the case of cone domains, as enabled by the following result.

\begin{proposition}\label{bilinear_form}
    Assume that $\cD$ is a cone and translation invariant. Then, solving the problem \eqref{weakproblem} with $\Omega = \R^d$ is equivalent to solving the following problem:
    \be\label{prbm:max_inner_prod_var_1}
        P_1 =    \max_{\substack{\mu \in \cD \cap \cP_{2,0}(\R^d),\, \pi \in \Pi(\mu, \nu)  \\ \Var(\mu) \leq 1 }} \int \langle x, y \rangle d \pi(x,y).
    \ee
\end{proposition}

\begin{proof}
By Theorem \ref{thm: kant equivalence to relaxed prob}, the problem \eqref{weakproblem} is equivalent to the problem $
        \min_{\mu \in \cD} {\mathcal W}_2^2(\mu, \nu)$, 
    which, by Lemma \ref{lemma:optimality_cone} and equation \eqref{eqn: projection} in the proof of Theorem \ref{thm: kant equivalence to relaxed prob}, is also equivalent to:
    \be\label{prbm:max_inner_prod_equal_sec_momt}
      P_2 =  \max_{\substack{\mu \in \cD \cap \cP_{2,0}(\R^d),\,  \pi \in \Pi(\mu, \nu) \\ \int \langle x, y \rangle d\pi(x,y) = \int |x|^2 d\mu}} \int \langle x, y \rangle d \pi(x,y).
    \ee
We will show that any optimizer $(\mu, \pi)$ for \eqref{prbm:max_inner_prod_equal_sec_momt} induces a solution to \eqref{prbm:max_inner_prod_var_1}, and conversely. For $\mu \in \cP(\R^d)$, $\pi \in \Pi(\mu,\nu)$ and $\la > 0$, define $\mu_\la = \la_\# \mu$ and $\pi_\la = (\la \times \id)_\# \pi \in \Pi(\mu_\la, \nu)$. Then for any feasible pair $(\hat \mu, \hat \pi)$ for \eqref{prbm:max_inner_prod_var_1} and $(\check\mu,\check\pi)$ for \eqref{prbm:max_inner_prod_equal_sec_momt} with $\int \langle x,y\rangle d \hat \pi >0$ and $ \check\mu \neq \delta_0$,
    \begin{align}\label{eqn:lambda2}
        &1 = \int |x|^2 d \check\mu_{\check\la} = {\check\la}^2 \int |x|^2 d\check \mu = {\check\la}^2 \int \langle x, y \rangle d\check \pi = {\check\la} \int \langle x, y \rangle d \check \pi_{\check\la}, \ \text{ and}\\
        \label{eqn:lambda1}
        &{\hat \la} \int \langle x, y \rangle d\hat \pi = \int \langle x, y \rangle d\hat \pi_{\hat \la} = \int |x|^2 d\hat \mu_{\hat \la} = {\hat \la}^2 \int |x|^2 d\hat \mu \le {\hat \la}^2,
    \end{align}
where $\check \la =  \sqrt{1/ \Var(\check\mu)}$ so that $\Var(\check\mu_{\check\la}) = 1$, while $\hat \la >0$ is the unique constant yielding the second equality in \eqref{eqn:lambda1}. This shows $(\check\mu_{\check\la}, \check \pi_{\check\la})$ is feasible for \eqref{prbm:max_inner_prod_var_1} and $(\hat \mu_{\hat \la} , \hat \pi_{\hat \la})$ for  \eqref{prbm:max_inner_prod_equal_sec_momt}, and moreover, if $P_1, P_2 >0$, then for any optimal pair $(\hat \mu, \hat \pi)$ for \eqref{prbm:max_inner_prod_var_1} and $(\check\mu,\check\pi)$ for \eqref{prbm:max_inner_prod_equal_sec_momt}, we have
\be\label{inverse}
\hat \la = 1 / \check \la,
\ee
since for any optimal pair \((\hat\mu, \hat\pi)\) for \eqref{prbm:max_inner_prod_var_1}, the constraint \(\Var(\hat\mu) \leq 1\) must be tight, meaning that the inequality in \eqref{eqn:lambda1} is satisfied as an equality. Then \eqref{inverse}, with \eqref{eqn:lambda2} and \eqref{eqn:lambda1}, shows that for any optimal $(\hat \mu, \hat \pi)$ for \eqref{prbm:max_inner_prod_var_1}, its scaling $(\hat \mu_{\hat \la} , \hat \pi_{\hat \la})$ is optimal for  \eqref{prbm:max_inner_prod_equal_sec_momt}, and conversely, for any optimal $(\check\mu,\check\pi)$ for \eqref{prbm:max_inner_prod_equal_sec_momt}, $(\check\mu_{\check\la}, \check \pi_{\check\la})$ is optimal for \eqref{prbm:max_inner_prod_var_1}. Finally, the proof also shows \(P_1 = 0\) if and only if \(P_2 = 0\), in which case \(\mu = \delta_0\) serves as the trivial solution to both problems.
\end{proof}

In the following, we apply this result to solve two examples with cone domains\footnote{Source code for examples \ref{ex_bounded_length_sd} and \ref{ex_bounded_curvature} are available at https://github.com/souza-m/data-denoising.}. In both cases, we use fixed weights for \(\mu\), specifically \(\mu = \frac{1}{m} \sum_{i=1}^{m} \delta_{x_i}\), to approximate a dataset of \(n = 300\) distributed along a step-shaped curve with added noise.

\subsection{Curves with bounded length-to-standard-deviation ratio}\label{ex_bounded_length_sd} 
Here we consider the following modification of the domain $\cD_1$:
\be
 \cD_2 = \bigg\{ \mu = \frac1m \sum_{i=1}^m  \delta_{x_i} \,\bigg|\, x_i \in \R^d,\,  \frac{L(\mu)}{{\rm SD}(\mu)}\leq B \bigg \}, \nn
\ee
where ${\rm SD}(\mu) = \sqrt{\Var(\mu)}$ represents the standard deviation of $\mu$. Thus, the bound is now imposed on the ratio between the length and the standard deviation. Since rescaling $\mu$ does not change the ratio $L(\mu) / {\rm SD}(\mu)$, we see that $\cD_2$ is a cone.

For each value of $B$, we calculate $m= 100$ points $x_1,...,x_{100}$ and linearly connect them. Figure \ref{bounded_ratio_curvature} illustrates the resulting curves for two values of $B$. 
We include an optimal solution under the convex order constraint, with size $m = 10$; higher values resulted in longer processing times or even out-of-memory error. Both methods seem to capture the underlying structure of the data reasonably well.
Supplementary Material C describes the alternating numerical steps for the solution.

\begin{figure}[h!]
   \centering 
    \makebox[\columnwidth]{\includegraphics[width=1\textwidth, trim=0 25 0 0]{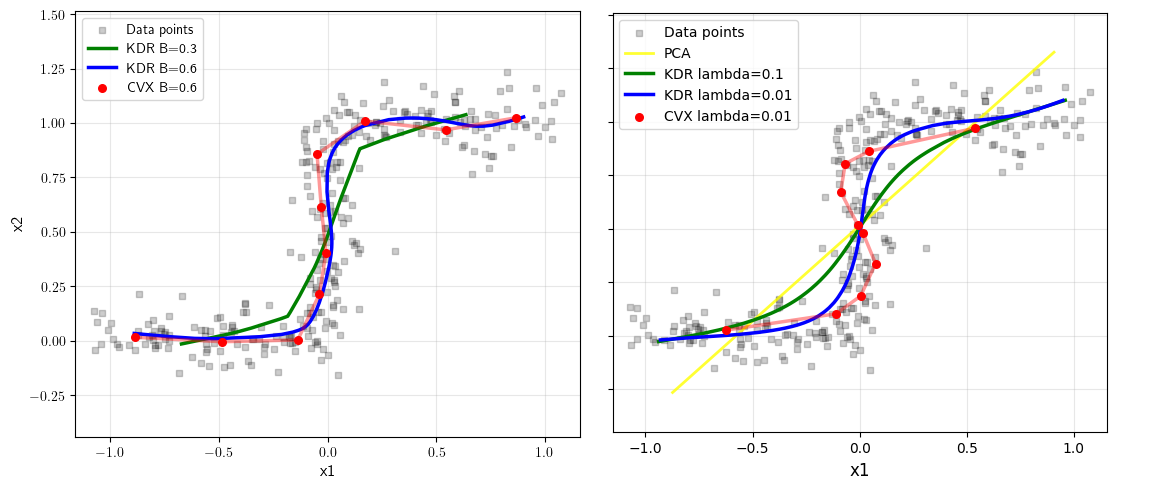}}
    \caption{Left: Curves with bounded length-standard deviation ratios. Right: Curves with bounded curvatures. Curves are formed by connecting points using straight lines.
    }
    \label{bounded_ratio_curvature}
\end{figure}

\subsection{Curves with bounded curvature}\label{ex_bounded_curvature}
Our last example considers a domain given by
\be
\cD_3 = \bigg\{ \mu = \frac1m \sum_{i=1}^m  \delta_{x_i} \,\bigg|\, x_i \in \R^d,\, \phi(\mu) \le B \bigg\}, \nn
\ee
where $\phi(\mu) = \sum_{i=2}^{m-1}\cos^{2}\frac{\theta_{i}}{2}$ represents the total curvature, with $\theta_{i}$ being the angle between segments $\overline{x_{i-1} x_i}$ and $\overline{x_i x_{i+1}}$. $\cD_2$ is a cone since the angles do not change with scaling of $\mu$, and by Lemma \ref{PCAlemma}, $B=0$ corresponds to the problem of finding the first principal direction.

In the numerical computation, the curvature constraint is handled indirectly via penalization and solved using an alternating method, as detailed in the Supplementary Material C. For each curvature penalty parameter $\la$, we set $m=100$ points and connect them linearly to form curves in Figure \ref{bounded_ratio_curvature}, where the first principal direction is also included for reference. A solution under the convex order constraint (with $m = 10$ as in Section \ref{ex_bounded_length_sd}) is also included. In this example, we observe a slight convergence of the latter towards the barycenter, which can be due to the stricter nature of the convex order or to numerical approximation errors.

The examples in this section demonstrate that the proposed method can efficiently compute the optimal measure under the Kantorovich dominance with large datasets.

%%%%%%%%%%%%%%%%%%%%%%%%%%%%%%%%%%%%%%%%%%%%%%
%% Support information, if any,             %%
%% should be provided in the                %%
%% Acknowledgements section.                %%
%%%%%%%%%%%%%%%%%%%%%%%%%%%%%%%%%%%%%%%%%%%%%%
\begin{acks}[Acknowledgments]
B.P. is pleased to acknowledge the support of Natural Sciences and Engineering Research Council of Canada Discovery Grant numbers  04658-2018 and 04864-2024. The work of J.H. and M.S. is in partial fulfillment of their doctoral degrees.
\end{acks}

%%%%%%%%%%%%%%%%%%%%%%%%%%%%%%%%%%%%%%%%%%%%%%
%% Funding information, if any,             %%
%% should be provided in the                %%
%% funding section.                         %%
%%%%%%%%%%%%%%%%%%%%%%%%%%%%%%%%%%%%%%%%%%%%%%
%\begin{funding}
%The first author was supported by NSF Grant DMS-??-??????.

%The second author was supported in part by NIH Grant ???????????.
%\end{funding}

\bibliographystyle{plain}
\bibliography{biblio}

\newpage
%%%%%%%%%%%%%%%%%%%%%%%%%%%%%%%%%%%%%%%%%%%%%%
%% Supplementary Material, including data   %%
%% sets and code, should be provided in     %%
%% {supplement} environment with title      %%
%% and short description. It cannot be      %%
%% available exclusively as external link.  %%
%% All Supplementary Material must be       %%
%% available to the reader on Project       %%
%% Euclid with the published article.       %%
%%%%%%%%%%%%%%%%%%%%%%%%%%%%%%%%%%%%%%%%%%%%%%

%\appendix
%\section{Solution comparison between the problem \ref{problem}  and the problem studied in \cite{kramkov2022optimal}}
\begin{supplement}\label{counterexample_Kramkov}
\stitle{A: Solution comparison between the problem \eqref{problem}  and the problem studied in \cite{kramkov2022optimal}.}
%\sdescription{Short description of Supplement A.}
\end{supplement}

%\centering
\begin{center}
    \begin{tikzpicture}[scale=0.65] 
    \begin{axis}[
        axis lines = middle,
        grid = both,
        xmin = -3.5, xmax = 2.5,
        ymin = -0.7, ymax = 5.5,
        width=8cm,  % 80% of the original 10cm width
        height=8cm, % 80% of the original 10cm height
        %legend pos=north west
        legend style={font=\footnotesize, at={(0.0,1.0)}, anchor=north west, %draw=none
        }
        %legend style={font=\small}  % Reduce the legend font size
    ]
    % Black points
    \addplot[
        only marks,
        mark=*,
        color=black,
    ]
    coordinates {
        (-3,3)
        (-1,5)
        (0,0)
        (2,2)
    };
    \addlegendentry{Support of $\nu$}

    % Blue points
    \addplot[
        only marks,
        mark=*,
        color=blue,
    ]
    coordinates {
        (-3/2, 3/2)
        (1/2, 7/2)
    };
    \addlegendentry{Support of $\mu_k$}

    % Green points
    \addplot[
        only marks,
        mark=*,
        color=green,
    ]
    coordinates {
        (-1/2, 1/2)
        (-1/2, 5/2)
        (-1/2, 9/2)
    };
    \addlegendentry{Support of $\mu_m$}

    \end{axis}
\end{tikzpicture}
\end{center}

We present an example which shows that the optimal measure for problem (1.2) in \cite{kramkov2022optimal} is not an optimal measure for our problem \eqref{problem}.
We define the measure $\mu_k, \mu_m$ and $\nu$ as follows:
\begin{align*}
\mu_k &= \tfrac{1}{2} \delta_{(-3/2, 3/2)} + \tfrac{1}{2} \delta_{(1/2, 7/2),}\\
\mu_m &= \tfrac{3}{10} \delta_{(-1/2, 1/2)} + \tfrac{2}{5} \delta_{(-1/2, 5/2)} + \tfrac{3}{10} \delta_{(-1/2, 9/2)},\\
\nu &= \tfrac{1}{4} \delta_{(0,0)} + \tfrac{1}{4} \delta_{(-3,3)} + \tfrac{1}{4} \delta_{(2,2)} + \tfrac{1}{4} \delta_{(-1,5)}.
\end{align*}

We define the martingale coupling \(\pi_k\) between $\mu_k$ and $\nu$ as follows:
\begin{equation*}
\pi_k = \tfrac{1}{4} \delta_{( (-3/2, 3/2), (0,0))} + \tfrac{1}{4} \delta_{((-3/2, 3/2), (-3,3))} + \tfrac{1}{4} \delta_{((1/2, 7/2), (2,2))} + \tfrac{1}{4} \delta_{((1/2, 7/2), (-1,5))},
\end{equation*}
while the martingale coupling \(\pi_m\) between $\mu_m$ and $\nu$ as follows:
\begin{align*}
\pi_m = &\, \tfrac{1}{4} \delta_{((-1/2, 1/2), (0,0))} + \tfrac{1}{20} \delta_{( (-1/2, 1/2), (-3,3))} + \tfrac{1}{5} \delta_{((-1/2, 5/2), (-3,3))} \\
       & + \tfrac{1}{4} \delta_{((-1/2, 9/2), (-1,5))} + \tfrac{1}{20} \delta_{((-1/2, 9/2), (2,2))} + \tfrac{1}{5} \delta_{((-1/2, 5/2), (2,2))}.
\end{align*}

For $u=(u_1,u_2)$, $v=(v_1,v_2) \in \R^2$, let $c_k(u,v) = (u_1 - v_1)(u_2 - v_2)$ be the cost function considered in \cite{kramkov2022optimal}. It is straightforward to check that for any $(x_0, y_0), (x_1, y_1) \in \spt(\pi_k)$,
\begin{equation}
    (1-t) c_k(x_0, y_0) + t c_k(x_1, y_1) \leq t (1-t) c_k(y_0, y_1), \qquad  t \in [0,1]. \label{Kram}
\end{equation}
By \cite[Theorem 2.2]{kramkov2022optimal}, \eqref{Kram} shows that $\pi_k$ is the optimal coupling for their problem. 
On the other hand, it is also easy to check that $\pi_m$ provide a better coupling for our problem \eqref{problem}, showing the optimal measure for problem (1.2) in \cite{kramkov2022optimal} is not the optimal measure for \eqref{problem}.
\\

\begin{supplement}\label{compute1}
\stitle{B: Computational details for Section \ref{ex_bounded_length}.}
%\sdescription{Short description of Supplement A.}
\end{supplement}
The Lagrangian we maximize for the example in Section \ref{ex_bounded_length} is the following:
\begin{align*}
    \mathcal{L}&(x, u, \lambda_1, \lambda_{1,i}, \lambda_2, \lambda_3) = \Big( \sum_i \|x_i\|^2 u_i - \| \sum_i x_i u_i \|^2 \Big) 
    - \lambda_1 \Big( \sum_i u_i - 1 \Big) + \sum_i \lambda_{1,i} u_i\\ 
    & - \lambda_2 \Big( \sum_i \|x_{i+1} - x_i\| - B \Big)
    - \lambda_3 \Big( \Wtq(\mu, \nu) - \Big( \sum_j \|y_j\|^2 /n - \sum_i \|x_i\|^2 u_i \Big) \Big),
\end{align*}
where $\sum_i \|x_i\|^2 u_i - \left\|\sum_i x_i u_i \right\|^2$ is the variance of $\mu$,  $\lambda_1 \in \R$ is the Lagrange multiplier for the probability constraint $\sum_i u_i = 1$, $\lambda_{1,i} \geq 0$ enforces $u_i \geq 0$, $\lambda_2 \geq 0$ enforces the length constraint $L(\mu) \leq B$, and $\lambda_3 \geq 0$ enforces the Kantorovich dominance constraint. 

The complementary slackness condition for each constraint ensures that the Lagrange multipliers only contribute when their respective constraints are active. Specifically:
\begin{align*}
\lambda_1 \Big( \sum_i u_i - 1 \Big) &= 0, \quad \lambda_{1,i} u_i = 0 \ \  \forall i, \quad \lambda_2 \Big(  B - \sum_{i=1}^{m-1} \|x_{i+1} - x_i\| \Big) = 0, \quad \text{and}\\
&\lambda_3 \Big( \sum_j \|y_j\|^2 /n - \sum_i \|x_i\|^2 u_i - \Wtq(\mu, \nu)\Big) = 0.
\end{align*}

Let $\pi$ be an optimal coupling corresponding to $\Wtq(\mu, \nu)$. The gradients with respect to the location $x = (x_i)_{i=1}^m \in (\R^d)^m$ and weight $u= (u_i)_{i=1}^m \in [0,1]^m$ are computed as follows.  \\
{\em Gradient with respect to $x$:}
\[
\frac{\partial \mathcal{L}}{\partial x_i} = 2x_i u_i - 2 \Big( \sum_j x_j u_j \Big) u_i - 2 \lambda_3 x_i u_i - 2 \lambda_3 \sum_j \pi_{ij} (x_i - y_j) + \lambda_2 \Big( \frac{x_{i+1} - x_i}{\|x_{i+1} - x_i\|} - \frac{x_i - x_{i-1}}{\|x_i - x_{i-1}\|} \Big).
\]
{\em Gradient with respect to $u$:}
\[
\frac{\partial \mathcal{L}}{\partial u_i} = \|x_i\|^2 - 2\Big( \sum_j x_j u_j \Big) x_i - \lambda_1 + \lambda_{1,i} - \lambda_3 \Big( \|x_i\|^2 + \sum_j \pi_{ij} \|x_i - y_j\|^2 \Big).
\]
The optimal transport plan $(\pi_{ij})$ with marginal $\mu$ using current $(x_i)$ and $(u_i)$ and given $\nu$. is calculated using the Sinkhorn algorithm for efficiency. 

The multipliers are updated using projected gradient descent to satisfy the  KKT conditions, with non-negativity of the multipliers enforced by projecting onto the feasible region. Specifically, each update step is given by:
\begin{align*}
    \lambda_1 &\leftarrow \lambda_1 + \eta_{\lambda_1} \Big( 1 - \sum_i \mu_i  \Big), \\
    \lambda_{1,i} &\leftarrow \max \Big( 0, \lambda_{1,i} - \eta_{\lambda_{1,i}} \mu_i \Big), \quad \forall i, \\
    \lambda_2 &\leftarrow \max  \Big( 0, \lambda_2 + \eta_{\lambda_2}  \Big(  \sum_{i=1}^{m-1} \|x_{i+1} - x_i\| - B  \Big)  \Big), \\
    \lambda_3 &\leftarrow \max  \Big(  0, \lambda_3 + \eta_{\lambda_3} \Big(  \Wtq(\mu, \nu) - \sum_j \|y_j\|^2 \nu_j + \sum_i \|x_i\|^2 \mu_i  \Big)  \Big).
\end{align*}
The projection operator $\max(0, \cdot)$ ensures that $\lambda_{1,i}, \lambda_2$ and $\lambda_3$ remain non-negative, in line with the KKT requirements.
\\

\begin{supplement}\label{compute2}
\stitle{C: Computational details for Sections \ref{ex_bounded_length_sd} and \ref{ex_bounded_curvature}. }
%\sdescription{Short description of Supplement B.}
\end{supplement}

To approximate a solution to \eqref{prbm:max_inner_prod_var_1}, we propose an alternating procedure that splits the problem into two subproblems, optimizing \((\mu, \pi)\) with respect to each variable separately. The variables are \(x = (x_1,..., x_m) \in (\mathbb{R}^d)^m\), where each \(x_i = (x_i^1,..., x_i^d)\), and \(\pi = (\pi_{ij}) \in \mathbb{R}_{\geq 0}^{m \times n}\). These examples use the constant weight on the points of $\mu$, that is, we set $\mu = \tfrac1m \sum_{i=1}^m \delta_{x_i}$. \(\pi\) represents a transport plan / coupling between the variable \(\mu\) and the  data \(\nu = \frac{1}{n} \sum_{j=1}^n \delta_{y_j}\), subject to the marginal constraints \(\sum_j \pi_{ij} = \frac1m \) for all \(i\), and \(\sum_i \pi_{ij} = \frac{1}{n}\) for all \(j\).

The optimization begins with an initial set \(x^0\) satisfying the  constraints $\frac1m \sum_{i}x_{i}=0$ and $\frac1m \sum_{i}\left\Vert x_{i}\right\Vert ^{2}\leq 1$ and the domain constraint (e.g., the set of PCA projections in both examples). At each iteration \(t \geq 1\), we perform the following two steps, repeating until convergence. The final output is the pair \((x, \pi)\), where \(x\) is rescaled by the factor \(\hat{\lambda}\) given in \eqref{inverse}.

\emph{Step 1.} Given $x$, find $\pi$ that solves 
\begin{align*}
\max_{\pi \in \Pi(1/m, 1/n)} \sum_{i}\sum_{j}\pi_{ij}\left\langle x_{i},y_{j}\right\rangle \quad 
\text{st.} \quad  \sum_{i}\sum_{j}\pi_{ij}x_{i}=0 \quad \text{and} \quad
\sum_{i}\sum_{j}\pi_{ij}\left\Vert x_{i}\right\Vert ^{2} \leq 1.
\end{align*}

\emph{Step 2.} Given $\pi$, set $\bar{y}=\pi^{t}y$ and find $x$ that solves 
\begin{align*}
\max_{x \in \cD} \sum_{i}\left\langle x_{i},\bar{y}_{i}\right\rangle \quad 
\text{st.} \quad \frac1m \sum_{i}x_{i}=0 \quad \text{and} \quad 
 \frac1m \sum_{i} \left\Vert x_{i}\right\Vert ^{2}\leq 1.
\end{align*}

In Step 1, note that the constraints on $\pi$ are satisfied by construction, since the marginal condition imposes $\sum_{i}\sum_{j}\pi_{ij}x_{i}=\sum_{i}\frac1m x_{i}=0$ and $\sum_{i}\sum_{j}\pi_{ij}\left\Vert x_{i}\right\Vert ^{2} =\sum_{i}\frac1m \left\Vert x_{i}\right\Vert ^{2}\leq 1$. Thus, the problem falls into the class of traditional optimal transport with fixed marginals. To solve it, in both examples we apply the Sinkhorn method, as noted in Subsection \ref{ex_bounded_length}. Then to solve Step 2 in Subsection \ref{ex_bounded_length_sd}, we explicitly state the domain constraint on $x$ as follows:
\begin{align*}
\max_{x \in (\R^d)^m}  \sum_{i=1}^m \left\langle x_{i},\bar{y}_{i}\right\rangle \quad
\text{st.} \quad \sum_{i=1}^m x_{i}=0, \ 
 \sum_{i=1}^m \left\Vert x_{i}\right\Vert ^{2}\leq1, \ 
 \sum_{i=1}^{m-1} \left\Vert x_{i+1} - x_{i}\right\Vert\leq B.
\end{align*}

Since the second constraint is clearly binding, together with the third constraints it implies that any solution to the above satisfies the domain constraint $\frac{\sum_{i=1}^{m-1} \left\Vert x_{i+1} - x_{i}\right\Vert}{\sum_{i=1}^m \left\Vert x_{i}\right\Vert ^{2}}\leq B$. We convert this problem into a second-order conic program, which can be solved efficiently by interior point methods -- see \cite{boyd2004convex} for an overview. To do that, we define the variable to be optimized as
\[
(x_1^1,\ldots,x_1^d,\ldots,x_m^1,\ldots,x_m^d,a_1,\ldots,a_{m-1})
\]
where the superscript in $x_i$ denotes the dimensional component and $a_1,\ldots,a_{m-1}$ is an auxiliary variable. The second-moment constraint on $x$ is replaced by the equivalent
\[
\sqrt{\sum_{i=1}^m \sum_{k=1}^d (x_i^k)^2} \leq 1.
\]

The constraints are equivalently rewritten as
\begin{align*}
x_{1}^{k}+\ldots+x_{m}^{k} & =0 &  & k=1,\ldots,d,\\
\left\Vert x_{1}^{1}\ldots x_{1}^{d}\ldots x_{m}^{1}\ldots x_{m}^{d}\right\Vert  & \leq1,\\
\left\Vert x_{i+1}-x_{i}\right\Vert  & \leq a_{i} &  & i=1,\ldots,m-1,\\
a_{1}+\ldots+a_{m-1} & =B.
\end{align*}

We solve the problem in this format using Python and CVXOPT package.

In Step 2 of Subsection \ref{ex_bounded_curvature}, we replace the domain constraint  $\phi\left(\mu\right)\leq B$ by a linear penalization together with another, inner iteration loop, as follows. Call $a_{i}=x_{i+1}-x_{i}$ for $i=1,\ldots,m-1$, and $e_{i}=\frac{a_{i}}{\left\vert a_{i}\right\vert}-\frac{a_{i-1}}{\left\vert a_{i-1}\right\vert}$ for $i=2,\ldots,m-1$. Notice that $\phi(\mu) = \frac14 \sum_{i=2}^{m-1}\left\Vert e_{i}\right\Vert ^{2}$. At iteration $t$, the problem is solved for
$\phi^{t}$ defined as
\begin{align*}
\phi^{0}\left(x\right) & =0\\
\phi^{t}\left(x\right)
& =\frac14\sum_{i=2}^{m-1}\left\langle e_{i}^{t},\frac{x_{i+i}^{t-1}-x_{i}}{{\left\vert a_i\right\vert}^{t-1}}-\frac{x_{i}-x_{i-1}^{t-1}}{{\left\vert a_{i-1}\right\vert}^{t-1}}\right\rangle
 =\sum_{i=2}^{m-1} \varphi_{i}^{t} x_i + \text{constant},
\end{align*}
where $\varphi_{i}^{t}=-\frac14e_{i}^{t}(\frac{1}{{\left\vert a_{i}\right\vert}^{t-1}}+\frac{1}{{\left\vert a_{i-1}\right\vert}^{t-1}})$. The penalized problem is written as
\begin{align*}
\max_{x \in (\R^d)^m} \sum_{i}\left\langle x_{i},c_{i}\right\rangle \quad 
\text{st.} \quad \sum_{i}x_{i}=0 \quad \text{and} \quad 
 \sum_{i}\left\Vert x_{i}\right\Vert ^{2}\leq m,
\end{align*}
where $c_i=\bar{y}_i-\lambda \varphi_{i}^{t}$ for some penalization multiplier $\lambda$. We solve it  through the Lagrangian
\begin{align*}
\mathcal{L}\left(x,\alpha,\beta\right) & =\sum_{k=1}^{d}\sum_{i=1}^{m}x_{i}^{k}c_{i}^{k}-\sum_{k=1}^{d}\sum_{i=1}^{m}\alpha^{k}x_{i}^{k}-\beta\Big[\sum_{k=1}^{d}\sum_{i=1}^{m}(x_{i}^{k})^{2}-m\Big].
\end{align*}

The first derivatives are 
\begin{align*}
\mathcal{L}_{x_{i}^{k}}\left(x,\alpha,\beta\right) & =c_{i}^{k}-\alpha^{k} -2 \beta x_{i}^{k},\\
\mathcal{L}_{\alpha^{k}}\left(x,\alpha,\beta\right) & =-\sum_{i=1}^{m} x_{i}^{k},\\
\mathcal{L}_{\beta}\left(x,\alpha,\beta\right) & =m-\sum_{k=1}^{d}\sum_{i=1}^{m} (x_{i}^{k})^{2}.
\end{align*}

We can assume that the second-moment condition is binding and $\beta>0$,
since otherwise it would be possible to increase the objective function.
The first order conditions imply
\begin{align*}
 & \left(i\right) \ c_{i}^{k}-\alpha^{k} -2\beta x_{i}^{k} =0 \implies  \hat{x}_{i}^{k}  =\frac{1}{2\beta}\Big(c_{i}^{k}-\alpha^{k}\Big) \quad \forall k,i, \\
 & \left(ii\right) \ \sum_{i=1}^{m} \hat{x}_{i}^{k}  =\frac{1}{2\beta}\Big(\sum_{i=1}^{m}c_{i}^{k}-\alpha^{k}\Big)=0  \implies \alpha^{k} =\sum_{i=1}^{m}c_{i}^{k} \quad \forall k,\\
 & \left(iii\right) \  \sum_{k=1}^{d}\sum_{i=1}^{m} (\hat{x}_{i}^{k})^{2}  =\frac{1}{4\beta^{2}}\sum_{k=1}^{d}\sum_{i=1}^{m} \Big(c_{i}^{k}-\alpha^{k}\Big)=m \\
 &\qquad \implies \beta  =\frac{1}{2}\sqrt{\frac1m \sum_{k=1}^{d}\sum_{i=1}^{m}\Big(c_{i}^{k}-\alpha^{k}\Big)^{2} }.
\end{align*}
Replacing $\alpha$ and $\beta$ gives $x$
as a function of $c$. Finally, we update $x^{t}$ partially at each iteration, as $x^{t}=\epsilon x_{\text{opt}}^{t}+\left(1-\epsilon\right)x^{t-1}$ where $x_{\text{opt}}^{t}$ solves the problem for $\phi^t$. The loop stops when the sequence $\left(x^{t}\right)$ converges to a fixed point $x^{*}$,
and we get
\[
\phi^{*}\left(\mu\right)=\frac14\sum_{i=2}^{m-1}\left\langle e_{i}^{*},\frac{x_{i+i}^{*}-x_{i}^{*}}{{\left\vert a_i\right\vert}^*}-\frac{x_{i}^{*}-x_{i-1}^{*}}{{\left\vert a_{i-1}\right\vert}^*}\right\rangle =\phi\left(\mu\right).
\]

\end{document}